\documentclass[journal]{IEEEtran}
\usepackage{amsmath}
\usepackage{amssymb}
\usepackage{amsthm}
\usepackage[justification=centering]{caption}
\usepackage{subcaption}
\usepackage[dvipsnames]{xcolor}
\usepackage{enumitem}
\usepackage{hyperref}
\usepackage{cite}
\usepackage{graphicx}
\usepackage{wrapfig}
\usepackage{tikz}
\usepackage{pgfplots}
\usetikzlibrary{patterns,arrows}
\pgfplotsset{compat=1.18}
\usepackage{pgfplotstable}
\usepackage{makecell}

 \pgfplotstableset{col sep=comma}

\hypersetup{
	colorlinks,
	linkcolor={red!50!black},
	citecolor={blue!50!black},
	urlcolor={blue!80!black}
}

\def\vs{\vspace{1\baselineskip}}

\newcommand{\scriptf}{\mathcal{F}}

\newcommand{\scriptr}{\mathcal{R}}

\newcommand{\om}{\(\Omega\)}
\newcommand{\reals}{\mathbb{R}} 
\newcommand{\naturals}{\mathbb{N}} 
\newcommand{\paren}[1]{\left(#1\right)} 
\newcommand{\E}{\mathbb{E}}
\DeclareMathOperator{\Var}{\mathrm{Var}}
\renewcommand{\P}{\mathbb{P}}

\newcommand{\matr}[1]{\boldsymbol{#1}}

\def\preceqdot{\mathrel{\preceq\kern-.5em\raise.22ex\hbox{\(\cdot\)}}}

\theoremstyle{plain}
\newtheorem{lem}{Lemma}
\newtheorem{thm}{Theorem}
\newtheorem{cor}{Corollary}

\author{Olle Abrahamsson, Danyo Danev and Erik G. Larsson 
	\thanks{The authors are with the Department of Electrical Engineering (ISY), Link{\"o}ping University, 58183 Link{\"o}ping, Sweden  (e-mail:  \{olle.abrahamsson, danyo.danev, erik.g.larsson\}@liu.se). This work was supported in part by ELLIIT and the KAW foundation.}
}

\title{Strong Convergence of a Random Actions Model in Opinion Dynamics}

\begin{document}
	\maketitle
	\begin{abstract}
		We study an opinion dynamics model in which each agent takes a random Bernoulli distributed action whose probability is updated at each discrete time step, and we prove that this model converges almost surely to consensus. We also provide a detailed critique of a claimed proof of this result in the literature. We generalize the result by proving that the assumption of irreducibility in the original model is not necessary. Furthermore, we prove as a corollary of the generalized result that the almost sure convergence to consensus holds also in the presence of a stubborn agent which never changes its opinion. In addition, we show that the model, in both the original and generalized cases, converges to consensus also in \(r\)th mean.
	\end{abstract}
	\section{Introduction}
	The study of opinion dynamics in social networks goes back several decades; for a review, see e.g. \cite{proskurnikov2017,proskurnikov2018,dong18}. For an overview of recent publications, see e.g. \cite{noorazar2020}. A class of opinion dynamics models of particular interest in relation to this paper incorporates randomness, for example in terms of random interactions \cite{acemoglu2010,mukhopadhyay2016}, or as in \cite{hunter2018}, where at each time \(t\) a randomly selected agent communicates a random opinion to its neighbors. The latter model also features the interesting novelty that an agent may grow increasingly stubborn over time.
	
	A majority of these dynamical models can be classified into to two categories, depending on if they address discrete-valued or continuous-valued opinions. Therefore, an important milestone in the literature was the CODA (continuous opinion and discrete action) model introduced in \cite{martins08}, which addressed the so-called community cleavage problem posed by Abelson \cite{abelson67}: If so many of the existing models leads to consensus, then why in society is there so much polarization around controversial issues? The CODA model allocates the agents' true, latent opinions to an unobservable continuous space, while the binary actions (e.g., \(0\) or \(1\)) which are observable by each agents' neighbors, are elements in a discrete space. An agent's latent opinion is its probability of taking action \(0\) or \(1\), and this probability is updated based on the observed actions of its neighbors. The updating scheme allows for the latent opinions to reach extreme values in the continuous space, leading to polarization and thereby resolving the cleavage problem.
	
	The CODA model has spawned a multitude of studies of different variations. A recent example is \cite{zino20} where the authors propose a two-layered network. In the first layer, the agents exchange and update their continuous opinions, and in the second layer, the agents observe their neighbors discrete binary actions. The purpose of introducing two layers is that an individual might choose to share his opinion with only his family and a few close friends, but is able to observe the actions from a different subset of his social network. Through this model the authors study the formation of e.g. unpopular norms. The model may help explain social collective behaviors in which a majority of individuals in a community can hold opinions that differ significantly from the actions taken by the community. Another example is \cite{jiao21} in which the CODA model is modified in two major ways: First, they take into account the fact that an agent's interactions are not limited to its neighbors (with online social media services as a striking example), and that an agent likely has some preference for which interactions to participate in. Second, they incorporate a game theoretic mechanism with the justification that the acquisition of an interaction is not necessarily the true observed action, but may be only extrinsic information that can be inferred.
	
	A recent variation of the CODA model which incorporates randomness was given in \cite{scaglione}. Under this setting, at every time step \(t\) each agent \(i\) chooses a Bernoulli distributed random action \(a_{t,i}~\sim~\mathrm{Bernoulli}(x_{t,i})\), and the corresponding update rule is
	\begin{equation}\label{eq:RA_update}
		\matr{x}_{t+1} = (1-\alpha)\matr{x}_t + \alpha\matr{W}\matr{a}_t,
	\end{equation}
	where \(\matr{x}_t\) represents the agents' opinions at time \(t\) and \(\matr{W}\) is an adjacency matrix that encodes the trust between agents. (This is explained in detail in Section \ref{sec:RA}.) In this model, which we will refer to as the Random Actions model (RA model for short), the probabilities of taking an action, rather than the actions/opinions themselves, are updated as a weighted average over the neighbors' actions. This is different from many classic opinion dynamics models such as that of DeGroot \cite{degroot} which assumes that the internal opinions are publicly known without taking into consideration the agents' actions.
		
	Other models from the recent literature that are closely related to the RA model include \cite{arieli21} in which the binary action \(a_{t,i}\) is chosen with a probability that does not match the opinion \(x_{t,i}\) exactly, but instead is a slightly exaggerated version of that opinion, and \cite{zhan21} in which the opinions of an agent's neighbors (as well as the actions of all agents) are taken into consideration in the updating scheme.
	
	Speaking more broadly, there are many models which are closely related to the RA model in that they all concern linear consensus dynamics with random interaction mechanism and/or random topologies. In \cite{fagnani07}, the time-variant sequence of weight matrices is a stochastic process. The authors show that the convergence rate for large scale networks is best described by mean square analysis. A necessary and sufficient condition for convergence to consensus with random weight matrices is given by \cite{tabhaz-salehi08}, which shows that under fairly mild conditions on the random matrices, the model reaches consensus almost surely. In \cite{dimakis10}, gossiping algorithms reaches consensus via pairwise interactions, which effectively results in a random topology. Similarly, both \cite{patterson10} and \cite{kar09} consider linear consensus models with links that fail at random, due to communication errors which, again, effectively results in a random topology. In \cite{patterson10}, the authors quantify the effect of such failures on the performance of the distributed gossip model, by characterizing the convergence rate to consensus. In \cite{kar09} the communication in the distributed gossip network is, in addition, assumed to be noisy, which leads to a trade-off between bias and variance in the limiting consensus. Two algorithms are proposed to mitigate this problem; one based on decaying time-variant weights, and the other based on Monte Carlo averaging.

	\section{Contributions}
	Much of the motivation for this work stems from the RA model in \cite{scaglione} and the claimed proof of Theorem 1 therein, in which all agents are asserted to converge almost surely to consensus. As detailed in Section \ref{sec:critique}, however, we have concerns with some of the arguments in \cite{scaglione}. The purported proof contains several mathematically ambiguous statements. Not only does this make it impossible to verify the proof without making extra assumptions, it is also not clear how to proceed past some of the steps. In addition, some of the arguments in the claimed proof are inaccurate.
	
	Our main contributions are a correct proof of \cite[Theorem 1]{scaglione} together with a generalized result which reveals that the irreducibility assumption in the original theorem is a sufficient but not a necessary condition. As far as we are aware, this is the first analysis of the RA model with a reducible network. Our proof strategy is different, and for the most part we derive the result from first principles by using basic definitions and properties in probability theory. In doing so, we hope that the argumentation in the proof becomes easier to follow.
	
	The remainder of the paper is organized as follows: In Section \ref{sec:model} we define the RA model and its underlying probability space. In Section \ref{sec:results} we provide a rigorous proof of the convergence result. For numerical simulations of the result, see \cite[Section IV A]{scaglione}. Then we provide a complete analysis of the generalized case where the network is reducible, and we prove under what conditions convergence to consensus still holds. In particular we prove that the irreducibility condition is sufficient but not necessary. As a special case we show that almost sure convergence to consensus holds also in the case when one agent is stubborn; that is, an agent who influences others but is never influenced itself. In addition, we show that if the RA model converges almost surely to consensus, then it also converges to consensus in \(r\)th moment, for all \(r > 0\). This is followed by a detailed critique of the proof of \cite[Theorem 1]{scaglione} in Section \ref{sec:critique}. The status of the proofs of various modes of convergence under different assumptions about the RA models are summarized in Table \ref{table:status}. In Section \ref{sec:discussion} we discuss considerations of a modified RA model with a time-variant weight matrix, as well as similarities and differences between the RA model and models in distributed optimization. Finally, we conclude the paper with a discussion about the results and their interpretations in Section \ref{sec:openproblems}.

	\begin{table}
		\caption{Status of proofs of modes of convergence (almost sure, in \(r\)th moment and in probability) towards consensus for the Random Actions model, with and without a stubborn/drifting agent, respectively.}
		\label{table:status}
		\begin{tabular}{c|c|c|c}
			 & a.s. & \(L_r\) & \(p\)\\
			 \hline
			 \textbf{RA}& \makecell{\cite{scaglione}, however \\see our critique\\ in Section \ref{sec:critique}.\\ See our proof \\ of Theorem \ref{thm:main}.}&\makecell{Corollary \ref{cor:RA-moments}}& \makecell{Follows from \\ Theorem \ref{thm:main}.}\\
			\hline
			\thead{RA with stubborn\\or drifting agent}&\makecell{Corollary \ref{cor:RA-as-stubborn}}&\makecell{Corollary \ref{cor:RA-moments}}&\makecell{Follows from \\ Corollary \ref{cor:RA-as-stubborn}.\\ See also our \\
				proof in \cite{abrahamsson2019}.}
		\end{tabular}
	\end{table}
	\section{Models and Definitions}\label{sec:model}
	In all models described in this section, we will consider a directed, weighted, single-component and strongly connected network with \(N\) nodes, where the nodes are interpreted as agents. Before giving the details of the model, let us at this point remind the reader that a \textit{row-stochastic} matrix is a square, non-negative matrix such that the row sums are equal to \(1\). The word ``row'' will be omitted and implied from hereon.
	
	\subsection{The RA Model}\label{sec:RA}
	In the RA model \cite[Equation (1)]{scaglione}, at every time step \(t\in \naturals = \{1,2,\dots\}\), each agent \(i~\in~\{1,2,\dots,N\}\) chooses one of two actions, \(0\) or \(1\), and these actions are generated by a Bernoulli random variable \(a_{t,i}\) with probability \(x_{t,i}\) of choosing action \(1\). The update of these probabilities is governed by \eqref{eq:RA_update}, where \(\matr{x}_1 \in\mathbb{R}^N\) is a column vector representing the initial opinions of the \(N\) agents, \(\alpha\in(0,1)\), the adjacency matrix \(\matr{W}\) is an \(N\times N\) stochastic matrix representing the trusts between agents, and 
	\begin{equation}
	\matr{a}_t = (a_{t,1}, a_{t,2}, \dots, a_{t,N})^T \in \{0,1\}^N
	\end{equation} are the actions with corresponding probabilities
	\begin{equation}
	\matr{x}_t = (x_{t,1}, x_{t,2}, \dots, x_{t,N})^T \in[0,1]^N,
	\end{equation} which themselves are random variables. We note that the only restriction on the initial distribution of \(\matr{x}_0\) is that each element \(x_{0,i}\) satisfies \(0 \leq x_{0,i} \leq 1\).
	
	We use the convention that \(w_{ij}>0\) represents an edge from \(j\) to \(i\) whose weight is equal to the trust that \(i\) puts in \(j\).  Note also that \(\matr{W}\) is irreducible since the network is strongly connected.

	Let the \(N\times 1\) vector with all entries equal to one be denoted by \(\mathbf{1} = \begin{pmatrix}
	1 & 1 & \dots & 1
	\end{pmatrix}^T\), and let \(\matr{\pi} = \begin{pmatrix}
	\pi_1 & \dots & \pi_n
	\end{pmatrix}^T\). Then we have
	\begin{equation}
		\matr{\pi}^T\matr{W} = \matr{\pi}^T, \qquad \matr{W}\mathbf{1} = \mathbf{1},
	\end{equation}
	i.e., \(\matr{\pi}^T\) and \(\matr{1}\) are the left and right eigenvectors of \(\matr{W}\) corresponding to the eigenvalue \(1\). By Perron-Frobenius theorem for irreducible and non-negative matrices (see, e.g., \cite[Theorem 8.4.4]{HorJoh}), we can choose \(\matr{\pi}\) such that \(\mathbf{1}^T\matr{\pi} = 1\) and \(\matr{\pi} \succ \matr{0}\). (That is, all elements of \(\matr{\pi}\) are positive and sum to one).
	\subsection{Probability model}
	We define the probability space \((\Omega',\scriptf',\P')\), with sample space
	\begin{equation}\label{eq:samplespace}
	\Omega' = [0,1]^N\times \paren{\{0,1\}^N}^{\aleph_0},
	\end{equation}
	where \(\scriptf'\) is the event space (\(\sigma\)-algebra) for \(\Omega'\), \(\P'\) is a probability measure and \(\aleph_0 = \lvert \naturals \rvert\).  We have
	\begin{equation}
	\begin{aligned}
	\Omega' = {} & \{(\matr{x}_1, \matr{a}_1, \matr{a}_2, \dots ) : \matr{x}_1 \in [0,1]^N, \matr{a}_t \in \{0,1\}^N, \\
	& t = 1,2,\dots\}.
	\end{aligned}
	\end{equation}
	It would be convenient to be able to reason about  \(\{\matr{x}_t\}_{t=1}^\infty\) as a sequence of random variables, but such sequences are not elements in \(\Omega'\). One way to bypass this obstacle is to utilize the opinion update rule \eqref{eq:RA_update}, since it maps actions \(\matr{a}_t\) in time step \(t\) to opinions \(\matr{x}_{t+1}\) in time step \(t+1\). Hence, we can construct a new probability space whose samples are all possible sequences \(\{\matr{x}_t\}_{t=1}^\infty\) (that are compatible with the model) if we define an appropriate mapping between the two spaces.
	
	In view of \eqref{eq:RA_update}, we define the mapping
	\begin{equation}
	\begin{aligned}
	h: \Omega' \to \paren{[0,1]^N}^{\aleph_0} = {} & \{(\matr{x}_1,\matr{x}_2,\dots)  : \matr{x}_t \in [0,1]^N, \\
	&t=1,2,\dots\},
	\end{aligned}
	\end{equation}
	by
	\begin{equation}
	h(\omega') = h((\matr{x}_1,\matr{a}_1,\dots)) = (\matr{z}_1,\matr{z}_2,\dots),
	\end{equation}
	where
	\begin{equation}
	\begin{aligned}
	\matr{z}_1 &= \matr{x}_1,\\
	\matr{z}_{t+1} &= (1-\alpha)\matr{z}_t + \alpha \matr{W}\matr{a}_t, \quad t=1,2,\dots
	\end{aligned}
	\end{equation}
	For convenience, we will henceforth write \(h(\omega') = (\matr{x}_1,\matr{x}_2,\dots)\) instead of \(h(\omega') = (\matr{z}_1,\matr{z}_2,\dots)\).
	Under this mapping, we in turn define a new probability space \((\Omega,\scriptf,\P)\), based on the original probability space \((\Omega',\scriptf',\P')\), by
	\begin{subequations}
		\begin{equation}
		\Omega = h(\Omega') = \{h(\omega') : \omega' \in 	\Omega'\} \subset \paren{[0,1]^N}^{\aleph_0},
		\end{equation}
		\begin{equation}
		\scriptf = h(\scriptf') = \{h(A) : A\in \scriptf'\},
		\end{equation}
		where
		\begin{equation}
		h(B) = \{h(b) : b\in B\},
		\end{equation}
		and
		\begin{equation}
		\P(B) = \P'(h^{-1}(B)), \quad \forall \ B \in \scriptf,
		\end{equation}
	\end{subequations}
	where \(h^{-1}(B) = \{\omega' \in \Omega': h(\omega')\in B\}\) denotes the preimage of \(B\) under \(h\).
	
	For any event \(A \in \scriptf'\), we define the joint cumulative distribution function (cdf) thusly: Let \(n,m \in \naturals\), and let
	\begin{equation}
	\begin{aligned}
	i &= (i_1,\dots,i_n) \in \naturals^n, \quad t = 	(t_1,\dots,t_n) \in \naturals^n\\
	j &= (j_1,\dots,j_m) \in \naturals^m, \quad u = (u_1,\dots, u_m) \in \naturals^m,
	\end{aligned}
	\end{equation}
	where the elements of \(i\) and \(j\) are agent indices and those of \(t\) and \(u\) are time indices. Then the \((n+m)\)-dimensional random variable
	\begin{equation}
		(x,a) = (x_{t_1,i_1},\dots,x_{t_n,i_n},a_{u_1,j_1},\dots, a_{u_m,j_m})
	\end{equation}
	has joint cdf \(F_{x,a}: [0,1]^n \times \{0,1\}^m \to [0,1]\), defined by
	\begin{equation}
	\begin{aligned}
		F_{x,a}(\matr{x},\matr{a}) = \P'(& h^{-1}(\{\omega \in \Omega : x_{t_k,i_k} \leq x_k, k=1,\dots,n\}) \\
		& \cap \{\omega' \in \Omega' : a_{u_\ell,j_\ell} \leq a_{\ell}, \ell = 1,\dots,m\}),
	\end{aligned}
	\end{equation}
	where \(\matr{x} = (x_1,\dots,x_n) \in [0,1]^n\) and \(\matr{a} = (a_1,\dots,a_m) \in \{0,1\}^m\).
	An example illustrates the construction:\\
	For \(n=m=1\) with \((x,a) = (x_{1,4},a_{2,3})\), we have
	\begin{equation}
	\begin{aligned}
	F_{x_{1,4},a_{2,3}}(\matr{x},\matr{a}) = \P'(&h^{-1}(\{\omega \in \Omega : x_{1,4} \leq x_1\})\\
	 &\cap \{\omega' \in \Omega' : a_{2,3} \leq a_1\}),
	\end{aligned}
	\end{equation}
	where \(x_{1,4},x_1 \in [0,1]\) and \(a_{2,3},a_1 \in \{0,1\}\).
	\section{Results}\label{sec:results}
	Our main result is that the RA model converges to consensus almost surely. The proof strategy is as follows: First we define events which contain undesirable sample paths, and show the the probabilities of these events tend to \(0\) as \(t \to \infty\). After a series of lemmas, whose proofs are found in the appendix, we see that the model converges to consensus in probability. The final push is a theorem which shows that the convergence to consensus holds also almost surely.
	
	With this strategy in mind, let us begin by defining a sequence of special subsets of the sample space \(\Omega\). Via lemmas \ref{lem:conv_bti} to \ref{lem:union_corners} we prove a series of useful facts about the probabilities of those subsets, where each lemma builds on the previous one.
	\vs
	
	First, let \(\delta \in (0,1/2)\), and consider the subsets \(C_{t,i}(\delta) \subset \Omega\) of the sample space, defined for each triple \((t,i,\delta)\) by
	\begin{equation}\label{eq:C_ti_definition}
	C_{t,i}(\delta) = \{(\matr{x}_1,\matr{x}_2,\dots) : x_{t,i} \in [0,\delta) \cup (1-\delta,1]\} \in \scriptf.
	\end{equation}
	The probability of these events converges to \(1\) in the limit as \(t  \to \infty\).
	\begin{lem}\label{lem:conv_bti}
		For any  \(\delta \in (0,1/2)\),
		\begin{equation*}
		\lim_{t\to\infty} \P\paren{C_{t,i}(\delta)} = 1, \quad i=1,\dots,N.
		\end{equation*}
	\end{lem}
	\hrule
	\vs
	We also consider the intersection of all events \(C_{t,i}(\delta)\) taken over all nodes in the network,
	\begin{equation}
		C_t(\delta) = \bigcap_{i=1}^N C_{t,i}(\delta),
	\end{equation}
	illustrated in Figure \ref{fig:Ct-delta}.
	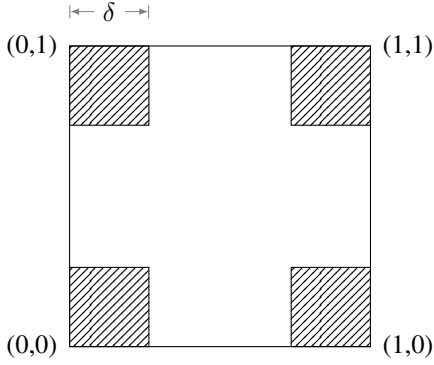
\begin{figure}[ht]
		\centering
		\begin{tikzpicture}[x=.5cm,y=0.5cm,>=latex]
		\def\RWd{8}
		\def\RHt{8}
		\def\CutSide{30pt}
		\draw
		(0,0) rectangle (\RWd,\RHt);
		\path[draw,pattern={north east lines} ] 
		(0,0)  rectangle ++(\CutSide,\CutSide) 
		(\RWd,0) rectangle ++(-\CutSide,\CutSide) 
		(0,\RHt) rectangle ++(\CutSide,-\CutSide) 
		(\RWd,\RHt) rectangle ++(-\CutSide,-\CutSide);
		
		\draw
		(-1,0)  node {(0,0)}
		(-1,\RHt)  node {(0,1)}
		(\RWd+1,0)  node {(1,0)}
		(\RWd+1,\RHt)  node {(1,1)};
		\begin{scope}[|<->|,help lines,text=black]
		\draw
		([yshift=13pt]0,\RHt) -- node[fill=white] {\(\delta\)} ++(\CutSide,0);   
		\end{scope}
		\end{tikzpicture}
		\caption{A sample path \(\omega = (\matr{x}_1,\matr{x}_2,\dots)\) belongs to \(C_t(\delta)\) if, at time \(t\), \(\matr{x}_t\) lies in the shaded region. The figure illustrates the two-dimensional case (\(N = 2\)).}
		\label{fig:Ct-delta}
	\end{figure}
	The following lemma shows that the probability of this intersection also converges to \(1\).
	\begin{lem}\label{lem:final}
		For any  \(\delta \in (0,1/2)\),
		\begin{equation*}
		\lim_{t\to\infty} \P\paren{C_t(\delta)} = 1.
		\end{equation*}
	\end{lem}

	\hrule
	\vs
	We define some new subsets of \om,
	\begin{equation}
	C^{(r)}_{t,i}(\delta) = 
	\begin{cases}
	\{(\matr{x}_1,\matr{x}_2,\dots)  \in \Omega : 0 \leq x_{t,i} < \delta\}, & r = 0\\
	\{(\matr{x}_1,\matr{x}_2,\dots)  \in \Omega : 1-\delta < x_{t,i} \leq 1\}, & r = 1,
	\end{cases}
	\end{equation}
	and we make the observation that these two sets partition \(C_{t,i}(\delta)\) for any \(\delta \in (0,1/2)\), i.e.,
	\begin{equation}
	C_{t,i}(\delta) = C^{(0)}_{t,i}(\delta) \cup C^{(1)}_{t,i}(\delta),
	\end{equation}
	and they are disjoint:
	\begin{equation}
	C^{(0)}_{t,i}(\delta) \cap C^{(1)}_{t,i}(\delta) = \emptyset.
	\end{equation}
	The set \(C_{t}(\delta)\) can therefore be rewritten as
	\begin{equation}\label{eq:Ct-delta}
	\begin{aligned}
	C_{t}(\delta) &= \bigcap_{i=1}^N C_{t,i}(\delta)\\
	&= \bigcap_{i=1}^N \paren{C^{(0)}_{t,i}(\delta) \cup C^{(1)}_{t,i}(\delta)}\\
	&= \bigcup_{\matr{m} \in \{0,1\}^N}\paren{\bigcap_{i=1}^N C^{(m_i)}_{t,i}(\delta)},
	\end{aligned}
	\end{equation}
	where  \(\matr{m} = (m_1,\dots,m_N) \in \{0,1\}^N\).
	
	With these new definitions at hand, we obtain the following crucial result, which states that if there is a directed edge between two agents, then the probability that they simultaneously take  different actions will decrease to zero in the limit as \(t \to \infty\).
	\begin{lem}\label{lem:twocorners}
		Suppose \(k,l\) are two agents such that \(w_{kl} > 0\). Then, for all \(\delta > 0\) satisfying \(\delta \leq \min\{\alpha w_{kl},\frac{1-\alpha}{2-\alpha}\}\), and for all \(\epsilon > 0\), there exists a time \(t_\epsilon\) such that for all \(t \geq t_\epsilon\),
		\begin{equation*}
		\P \paren{C_{t,k}^{(0)}(\delta) \cap C_{t,l}^{(1)}(\delta)} \leq \delta + \epsilon,
		\end{equation*}
		and
		\begin{equation*}
			\P \paren{C_{t,k}^{(1)}(\delta) \cap C_{t,l}^{(0)}(\delta)} \leq \delta + \epsilon.
		\end{equation*}
	\end{lem}
	\vs
	\hrule
	\vs
	Finally, we define the event
	\begin{equation}
	C_t^{\matr{m}}(\delta) = \bigcap_{i=1}^N C_{t,i}^{(m_i)}(\delta), \quad \matr{m} \in \{0,1\}^N.
	\end{equation}
	Each binary vector \(\matr{m} \in \{0,1\}^N\) represents a corner in the \(N\)-dimensional unit cube, so \(C_t^{\matr{m}}(\delta)\) is the event that \(\matr{x}_t\) is \(\delta\)-close to \(\matr{m}\) at time \(t\). This is illustrated in Figure \ref{fig:Ct^0-delta} for \(\matr{m} = \mathbf{0}\).
	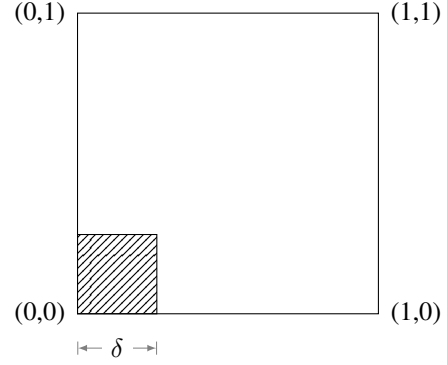
\begin{figure}[ht]
		\centering
		\begin{tikzpicture}[x=.5cm,y=0.5cm,>=latex]
		\def\RWd{8}
		\def\RHt{8}
		\def\CutSide{30pt}
		\draw
		(0,0) rectangle (\RWd,\RHt);
		\path[draw,pattern={north east lines} ] 
		(0,0)  rectangle ++(\CutSide,\CutSide) ;
		\draw
		(-1,0)  node {(0,0)}
		(-1,\RHt)  node {(0,1)}
		(\RWd+1,0)  node {(1,0)}
		(\RWd+1,\RHt)  node {(1,1)};
		\begin{scope}[|<->|,help lines,text=black]
		\draw
		([yshift=-13pt]0,0) -- node[fill=white] {\(\delta\)} ++(\CutSide,0);   
		\end{scope}
		\end{tikzpicture}
		\caption{A sample path \(\omega = (\matr{x}_1,\matr{x}_2,\dots)\) belongs to \(C_t^\mathbf{0}(\delta)\) if, at time \(t\), \(\matr{x}_t\) lies in the shaded corner. The figure illustrates the two-dimensional case (\(N = 2\)).}
		\label{fig:Ct^0-delta}
	\end{figure}
	In light of \eqref{eq:Ct-delta}, we obtain
	\begin{equation}
	C_t(\delta) = \bigcup_{m\in \{0,1\}^N} C_t^{\matr{m}}(\delta).
	\end{equation}
	The next lemma shows that the probability mass eventually is concentrated in the special corners \(\mathbf{0}\) and \(\mathbf{1}\). In other words, the RA model converges to consensus in probability.
	\begin{lem}\label{lem:union_corners}
		For all \(\epsilon >0\) and all \(\delta > 0\) there exists a time \(t_{\epsilon,\delta}\) such that for all \(t \geq t_{\epsilon,\delta}\),
		\begin{equation*}
		\P\paren{C_t^{\matr{0}}(\delta)} + \P\paren{C_t^{\matr{1}}(\delta)}  > 1 - \epsilon.
		\end{equation*}
	\end{lem}
	\vs
	\hrule
	\vs
	For the proof of our main result, we also need the following technical lemma (which does not rely on any of the previous discussion).
	\begin{lem}\label{lem:prod-cont-decr}
		Let \(\alpha \in (0,1)\), and for any \(S,N\in\naturals\), consider the product
		\begin{equation*}
			\prod_{s=0}^{S} (1-(1-\alpha)^s \gamma)^N.
		\end{equation*}
		In the limit as \(S \to \infty\), the product converges uniformly to a continuous and decreasing function \(g_{\alpha,N}(\gamma)\) on the closed interval \([0,1]\), with \(g_{\alpha,N}(0) = 1\) and \(g_{\alpha,N}(1) = 0\).
	\end{lem}
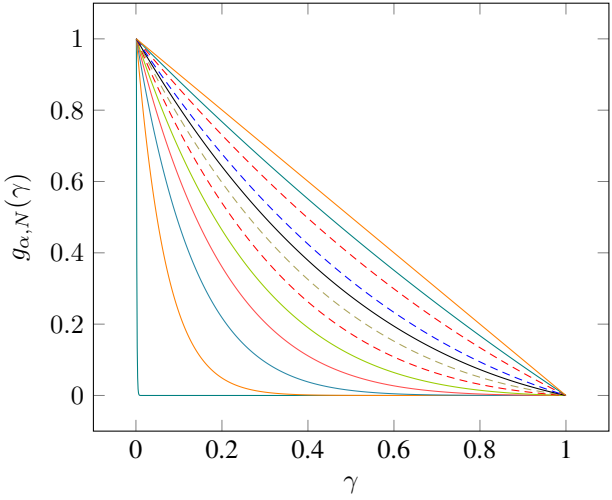
\begin{figure}[ht]
	\begin{tikzpicture}
		\begin{axis}[
			cycle list name=exotic,
			xlabel={\(\gamma\)},
			ylabel={\(g_{\alpha,N}(\gamma)\)},
			xtick =  {0, 200,..., 1000},
			xticklabels={0,0.2,0.4,0.6,0.8,1}
			]
			\pgfplotstableread{qpoch.dat}\datatable
			\addplot +[thin, mark = none] table[x index=0, y index = 1] {\datatable};
			\addplot +[thin, mark = none] table[x index=0, y index=3] {\datatable};
			\addplot +[thin, mark = none] table[x index=0, y index=5] {\datatable};
			\addplot +[thin, mark = none] table[x index=0, y index=7] {\datatable};
			\addplot +[thin, mark = none] table[x index=0, y index=9] {\datatable};
			\addplot +[thin, mark = none] table[x index=0, y index=11] {\datatable};
			\addplot +[thin, mark = none] table[x index=0, y index=13] {\datatable};
			\addplot +[thin, mark = none] table[x index=0, y index=15] {\datatable};
			\addplot +[thin, mark = none] table[x index=0, y index=17] {\datatable};
			\addplot +[thin, mark = none] table[x index=0, y index=21] {\datatable};
			\addplot +[thin, mark = none] table[x index=0, y index=25] {\datatable};
			\addplot +[thin, mark = none] table[x index=0, y index=30] {\datatable};
		\end{axis}
	\end{tikzpicture}
	\caption{The function \(g_{\alpha,N}(\gamma)\) from Lemma \ref{lem:prod-cont-decr}, for \(N = 6\) and where \(\alpha\)  is one of \(12\) evenly spaced values from \(0.001\) to \(0.999\). The smaller \(\alpha\) is, the faster the function decreases towards \(0\).}\label{fig:qpochhammer}
\end{figure}
	The function \(g_{\alpha,N}(\gamma)\) is illustrated in Figure \ref{fig:qpochhammer} on the interval \([0,1]\), for different values of \(\alpha \in (0,1)\).
	%
	
	We are now ready to prove our main result,  which states that the probability of a sample path being \(\delta\)-close (with \(\delta\) arbitrarily small) to either of the corners \(\mathbf{0}\) and \(\mathbf{1}\) beyond some point in the far enough future, can be made \(\epsilon\)-close to \(1\). In other words, in the limit as \(S\) tends to infinity, this will happen almost surely. In particular, it holds that \(\matr{x}_t \xrightarrow{a.s.} \{\matr{0},\matr{1}\}\).
	\begin{thm}\label{thm:main}
		For all  \(\epsilon > 0\) and all \(\delta > 0\), there exists a time instant \(t_{\epsilon,\delta}\in\naturals\) such that for every \(S \in \naturals\) we have
		\begin{equation}\label{eq:final}
			\P\paren{\bigcap_{s = 0}^S C_{t_{\epsilon,\delta}+s}^{\matr{0}}(\delta)} + \P\paren{\bigcap_{s = 0}^S C_{t_{\epsilon,\delta}+s}^{\matr{1}}(\delta)} > 1 - \epsilon.
		\end{equation}
	\end{thm}
	\begin{proof}
		If \(S = 0\), the theorem is simply Lemma \ref{lem:union_corners}. In the sequel we therefore assume that \(S \geq 1\). The first step of the proof is to show that for all such \(S\),
		\begin{equation}\label{eq:thm1-roadmap}
		\begin{aligned}
		&\P\paren{\bigcap_{s = 0}^S C_{t+s}^{\matr{0}}((1-\alpha)^s\delta)} + \P\paren{\bigcap_{s = 0}^S C_{t+s}^{\matr{1}}((1-\alpha)^s\delta)}\\
		& > \Big(\P(C_{t}^{\matr{0}}(\delta)) + \P(C_{t}^{\matr{1}}(\delta))\Big) A(\delta),
		\end{aligned}
		\end{equation}
		where \(A(\delta)\) is a deterministic function of \(\delta\). Then we use Lemma \ref{lem:prod-cont-decr} to show that \(A(\delta)\) can be made arbitrarily close to \(1\), depending only on the choice of \(\epsilon\). Further, we use Lemma \ref{lem:union_corners} to show that there exists a time \(t_0\) so that \(\P(C_{t}^{\matr{0}}(\delta)) + \P(C_{t}^{\matr{1}}(\delta))\) also can be made arbitrarily close to \(1\), for all \(t \geq t_0\). Finally we combine the results to show that the right hand side of \eqref{eq:thm1-roadmap} can be made arbitrarily close to \(1\), depending only on the choice of \(\epsilon\) and \(\delta\). Since
		\begin{equation}
		\label{eq:inclusion}
		\bigcap_{s = 0}^S C_{t+s}^{\matr{m}}(\delta) \supseteq \bigcap_{s = 0}^S C_{t+s}^{\matr{m}}((1-\alpha)^s\delta), \ \matr{m}\in \{\matr{0},\matr{1}\}, 
		\end{equation}
		the theorem then follows by re-indexing.
		\vs
		
		For the first step, we make the following observations, which hold for general \(t\): Let \(\delta > 0\). Suppose that \(\matr{x}(t)\) lies in the \(\mathbf{0}\) corner of size \(\delta\) and all agents take action \(0\), i.e., \(\omega \in C_t^{\matr{0}}(\delta)\) and \(a(t) = \mathbf{0}\). Then \(\matr{x}(t+1)\) stays in that corner but contracted by a factor \((1-\alpha)\). For each agent, the probability of taking action \(0\) is at least \(1-\delta\), and the actions are all independent. Note that
		\begin{equation}
		\P(h(\{\omega' \in \Omega' : a(t) = \mathbf{0}\}) \mid \{\omega: x(t) = \matr{x}\}) = \prod_{i=1}^N (1-x_i),
		\end{equation}
		and since we assume that \(\omega \in C_t^\mathbf{0}(\delta)\), it holds that \(\mathbf{0} \preceq \matr{x}(t) \prec \delta \mathbf{1}\), so
		\begin{equation}
		\prod_{i=1}^N (1-x_i) > (1-\delta)^N.
		\end{equation}
		Therefore, we obtain
		\begin{equation}\label{eq:thm1-firstbound}
		\P\paren{C_{t+1}^{\mathbf{0}}((1-\alpha)\delta) \mid C_t^{\mathbf{0}}(\delta)} > (1-\delta)^N.
		\end{equation}
		The Markov property of the model gives that
		\begin{equation}
		\begin{aligned}
		&\P\paren{C_{t+2}^{\matr{0}}((1-\alpha)^2\delta) \cap C_{t+1}^{\matr{0}}((1-\alpha)\delta) \mid C_t^{\matr{0}}(\delta)}\\
		= {} &\P\paren{C_{t+2}^{\matr{0}}((1-\alpha)^2\delta) \mid C_{t+1}^{\matr{0}}((1-\alpha)\delta) \cap C_t^{\matr{0}}(\delta)}\\
		\cdot & \P\paren{ C_{t+1}^{\matr{0}}((1-\alpha)\delta) \mid C_t^{\matr{0}}(\delta)}\\
		= {} & \P\paren{C_{t+2}^{\matr{0}}((1-\alpha)^2\delta) \mid C_{t+1}^{\matr{0}}((1-\alpha)\delta)}\\
		\cdot & \P\paren{ C_{t+1}^{\matr{0}}((1-\alpha)\delta) \mid C_t^{\matr{0}}(\delta)}.
		\end{aligned}
		\end{equation}
		Repeated application of this property together with inequality \eqref{eq:thm1-firstbound} results in
		\begin{equation}
		\P\paren{\bigcap_{s = 1}^S C_{t+s}^{\matr{0}}((1-\alpha)^s\delta) \mid C_t^{\matr{0}}(\delta)} > \prod_{s=0}^{S-1}(1-(1-\alpha)^s \delta)^N,
		\end{equation}
		from which it follows that
		\begin{equation}
		\begin{aligned}
		&\P\paren{\bigcap_{s = 0}^S C_{t+s}^{\matr{0}}((1-\alpha)^s\delta)}\\
		 = {} & \P\paren{\bigcap_{s = 1}^S C_{t+s}^{\matr{0}}((1-\alpha)^s\delta) \mid C_t^{\matr{0}}(\delta)}\P\paren{C_t^{\matr{0}}(\delta)}\\
		> {} & \P\paren{C_t^{\matr{0}}(\delta)} \prod_{s=0}^{S-1}(1-(1-\alpha)^s \delta)^N.
		\end{aligned}
		\end{equation}
		A similar derivation for the corner \(\matr{1}\) shows that
		\begin{equation}
		\P \paren{\bigcap_{s = 0}^S C_{t+s}^{\matr{1}}((1-\alpha)^s \delta)}
		> \P(C_t^{\matr{1}}(\delta))\prod_{s=0}^{S-1} (1-(1-\alpha)^s\delta)^N,
		\end{equation}
		and hence we obtain
		\begin{equation}\label{eq:thm1-bound}
		\begin{aligned}
		&\sum_{\matr{m}\in \{\matr{0},\matr{1}\}} \! \! \P\paren{\bigcap_{s = 0}^S C_{t+s}^{\matr{m}}((1-\alpha)^s\delta)}\\
		> {} & \sum_{\matr{m}\in \{\matr{0},\matr{1}\}} \! \! \P(C_t^{\matr{m}}(\delta))  \prod_{s=0}^{S-1} (1-(1-\alpha)^s\delta)^N.
		\end{aligned}
		\end{equation}
		This concludes the first step of the proof. For the remaining step, we deal first with the individual factors of the right hand side of \eqref{eq:thm1-bound}, and then we combine the results.
		\vs
		
		For the second factor (i.e., the finite product),  first note that for all \(\alpha \in (0,1)\), \(\gamma \in [0,1]\) and \(s \in \naturals\) we have \(0 < 1-(1-\alpha)^s \gamma\leq 1\), which leads to
		\begin{equation}
		\prod_{s=0}^{S-1} (1-(1-\alpha)^s \gamma)^N  \geq \prod_{s=0}^\infty (1-(1-\alpha)^s \gamma)^N.
		\end{equation}
		Let \(\epsilon > 0\). W.l.o.g., we can assume that \(\epsilon\leq 1\) (otherwise the result of the theorem is trivial). By Lemma \ref{lem:prod-cont-decr}, the infinite product converges to a continuous and strictly decreasing function \(g_{\alpha,N}(\gamma)\) on \([0,1]\), with \(g_{\alpha,N}(0) = 1\), so we can find a sufficiently small \(\delta_0 > 0\) such that
		\begin{equation}
		\label{eq:bound-second-part}
		\prod_{s=0}^{S-1} (1-(1-\alpha)^s \delta_0)^N > 1-\frac{\epsilon}{2-\epsilon} = \frac{2(1-\epsilon)}{2-\epsilon}  
		\end{equation}
		is valid for all \(S \in \naturals\). 
		
		For the first factor we apply Lemma \ref{lem:union_corners} with the parameters \(\epsilon/2\) and \(\mbox{min}\{\delta,\delta_0\}\). Then there exists a time \(t_{\epsilon,\delta}\) such that, for all \(t \geq t_{\epsilon,\delta}\), we have
		\begin{equation}
		\label{eq:bound-first-part}
		\sum_{\matr{m}\in \{\matr{0},\matr{1}\}} \! \!
		\P\paren{C_{t}^{\matr{m}}(\mbox{min}\{\delta,\delta_0\})} > 1- \frac{\epsilon}{2},
		\end{equation}
		and in particular it holds for \(t = t_{\epsilon,\delta}\).
		\vs
		
		Now we consider two possible cases. The first is \(\delta\leq\delta_0\). By combining the results in \eqref{eq:thm1-bound}, \eqref{eq:bound-second-part} and \eqref{eq:bound-first-part}, we obtain  
		\begin{equation}
		\begin{aligned}		
		&\sum_{\matr{m}\in \{\matr{0},\matr{1}\}}  \! \!
		\P\paren{\bigcap_{s = 0}^S C_{t_{\epsilon,\delta}+s}^{\matr{m}}((1-\alpha)^s\delta)}\\
		> {} & 	\sum_{\matr{m}\in \{\matr{0},\matr{1}\}} \! \! \P(C_t^{\matr{m}}(\delta)) \prod_{s=0}^{S-1} (1-(1-\alpha)^s\delta)^N\\
		\geq {} & \sum_{\matr{m}\in \{\matr{0},\matr{1}\}} \! \! \P(C_t^{\matr{m}}(\delta)) \prod_{s=0}^{S-1} (1-(1-\alpha)^s\delta_0)^N\\
		> {} & \paren{1- \frac{\epsilon}{2}} \frac{2(1-\epsilon)}{2-\epsilon}  = 1-\epsilon.
		\end{aligned}
		\end{equation}
		Here the second inequality is due to the fact that \(\prod_{s=0}^{S-1} (1-(1-\alpha)^s \gamma)^N\) is decreasing as a function of \(\gamma\) (on the interval \([0,1]\)).
		
		The second case is \(\delta > \delta_0\). Now we use the set inclusion
		\begin{equation}
		C_{t_{\epsilon,\delta}+s}^{\matr{m}}((1-\alpha)^s\delta)
		\supseteq 
		C_{t_{\epsilon,\delta}+s}^{\matr{m}}((1-\alpha)^s\delta_0),
		\end{equation}
		and apply a version of \eqref{eq:thm1-bound} where \(\delta\) is replaced by \(\delta_0\). Together with \eqref{eq:bound-second-part} and \eqref{eq:bound-first-part}, this results in   
		\begin{equation}
		\begin{aligned}
		&\sum_{\matr{m}\in \{\matr{0},\matr{1}\}} \! \!
		\P\paren{\bigcap_{s = 0}^S C_{t_{\epsilon,\delta}+s}^{\matr{m}}((1-\alpha)^s\delta)}\\
		\geq {} & \sum_{\matr{m}\in \{\matr{0},\matr{1}\}}  \! \!
		\P\paren{\bigcap_{s = 0}^S C_{t_{\epsilon,\delta}+s}^{\matr{m}}((1-\alpha)^s\delta_0)}\\ 
		> {} &
		\sum_{\matr{m}\in \{\matr{0},\matr{1}\}} \! \! \P(C_t^{\matr{m}}(\delta_0)) \prod_{s=0}^{S-1} (1-(1-\alpha)^s\delta_0)^N\\
		> {} & \paren{1- \frac{\epsilon}{2}} \frac{2(1-\epsilon)}{2-\epsilon}  = 1-\epsilon.
		\end{aligned}
		\end{equation}
		Hence, for any $\delta >0$ we have 
		\begin{equation}
		\label{eq:gencase}
		\sum_{\matr{m}\in \{\matr{0},\matr{1}\}} \! \! \P\paren{\bigcap_{s = 0}^S C_{t_{\epsilon,\delta}+s}^{\matr{m}}((1-\alpha)^s\delta)}\\ > 1-\epsilon.
		\end{equation}
		Finally, equations \eqref{eq:inclusion} and \eqref{eq:gencase} lead us to the desired result
		\begin{equation}
		\begin{aligned}
		&\sum_{\matr{m}\in \{\matr{0},\matr{1}\}} \! \!
		\P\paren{\bigcap_{s = 0}^S C_{t_{\epsilon,\delta}+s}^{\matr{m}}(\delta)}\\
		\geq {} & \sum_{\matr{m}\in \{\matr{0},\matr{1}\}} \! \!
		\P\paren{\bigcap_{s = 0}^S C_{t_{\epsilon,\delta}+s}^{\matr{m}}((1-\alpha)^s\delta)} > 1-\epsilon.
		\end{aligned}
		\end{equation}
	\end{proof}

	\subsection{Consensus results for reducible networks}
	We have just proved that a sufficient condition for the RA model to converge almost surely to consensus is that the network is strongly connected, i.e., \(\matr{W}\) is irreducible. However, this is not a necessary condition, as we will now show. Inspired by Markov chain terminology, we introduce a relation \(\scriptr\) and say that two agents \(i\) and \(j\) \textit{communicate}, denoted by \(i \sim_\scriptr j\), if there exists a directed path from \(i\) to \(j\), and a directed path from \(j\) to \(i\). (Algebraically this means that there exist non-negative integers \(n_1,n_2\) such that \([\matr{W}^{n_1}]_{ij}[\matr{W}^{n_2}]_{ji} \neq 0\), where \(n_1,n_2\) are the lengths of the two paths.)
	
	
	It is easy to see that \(\scriptr\) is an equivalence relation. Hence, it induces a set of equivalence classes \(C_1,\dots,C_M\), for some integer \(M \leq N\), which are the strongly connected components in the network (corresponding to \textit{communication classes} in a Markov chain). For two such components, say \(C_r\) and \(C_s\), there is either zero, one or several directed edges from one of them to the other. If there is an edge from, say, \(C_s\) to \(C_r\), then there is no edge in the other direction since otherwise the agents \(C_s\) and \(C_r\) would belong to the same strongly connected component.
	
	We define a relation \(\preceq\) on the collection \(\{C_1,\dots,C_M\}\), where \(C_r \preceq C_s\) means that there is at least one directed path from an agent in \(C_s\) to an agent in \(C_r\) (possibly going through other intermediate components). It is easy to verify that \(\preceq\) is reflexive, antisymmetric and transitive, so \(\{C_1,\dots,C_M\}\) is a poset (partially ordered set) under \(\preceq\). We say that \(C_r\) is a \textit{minimal element} in the poset if there is no \(C_q\) such that \(C_q \preceq C_r\), and we say that \(C_s\) is a \textit{maximal element} if there is no \(C_t\) such that \(C_s \preceq C_t\). 
	
	As an example of this construction, consider the network in Figure \ref{fig:RA-smallnetwork}. Clearly it has four strongly connected components:
	\begin{equation}
		\begin{aligned}
			C_1 &= \{v_1\}, & C_2 &= \{v_2,v_3\}, \\
			C_3 &= \{v_4,v_5\}, & C_4 &= \{v_6,v_7\}.
		\end{aligned} 
	\end{equation}
	Here, \(C_4 \preceq C_2 \preceq C_1\) and \(C_4 \preceq C_3\), so in this poset \(C_1,C_3\) are maximal elements and \(C_4\) is a minimal element.
	
	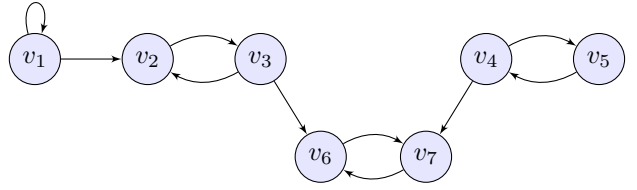
\begin{figure}[t]
		\centering
		\begin{tikzpicture}[auto,main_node/.style={circle,fill=blue!10,draw,minimum size=.5em,inner sep=3pt}, edge/.style={->,> = latex'}]
			
			\node[main_node] (1) at (-4.5,0)  {\(v_1\)};
			\node[main_node] (2) at (-3,0)  {\(v_2\)};
			\node[main_node] (3) at (-1.5, 0)  {\(v_3\)};
			\node[main_node] (4) at (1.5,0)  {\(v_4\)};
			\node[main_node] (5) at (3,0){\(v_5\)};
			\node[main_node] (6) at (-0.7,-1.3){\(v_6\)};
			\node[main_node] (7) at (0.7,-1.3){\(v_7\)};
			
			\draw[edge] (1) to [loop above] (1) {}; 
			\draw[edge] (1) to (2) {};
			\draw[edge] (2) to [bend left] (3) {};
			\draw[edge] (3) to [bend left] (2) {};
			\draw[edge] (3) to (6) {};
			\draw[edge] (4) to [bend left] (5) {};
			\draw[edge] (4) to (7) {};
			\draw[edge] (5) to [bend left] (4) {};
			\draw[edge] (6) to [bend left] (7) {};
			\draw[edge] (7) to [bend left] (6) {};
		\end{tikzpicture}%
		\caption[Short description]{A network with four strongly connected components.}\label{fig:RA-smallnetwork}
	\end{figure}
	\begin{figure}[t]
		\centering
		\begin{tikzpicture}[auto,main_node/.style={circle,fill=blue!10,draw,minimum size=.5em,inner sep=3pt}, edge/.style={->,> = latex'}]
			
			\node[main_node] (1) at (-2,0)  {\(C_1\)};
			\node[main_node] (2) at (-1,-1)  {\(C_2\)};
			\node[main_node] (3) at (1, 0)  {\(C_3\)};
			\node[main_node] (4) at (0,-2)  {\(C_4\)};
			
			\draw[edge] (1) to (2) {};
			\draw[edge] (2) to (4) {};
			\draw[edge] (3) to (4) {};
		\end{tikzpicture}%
		\caption[Short description]{Hasse diagram of the poset of strongly connected components obtained from the network in Figure \ref{fig:RA-smallnetwork}, with maximal elements \(C_1,C_3\) and minimal element \(C_4\).}\label{fig:RA-hasse}
	\end{figure}
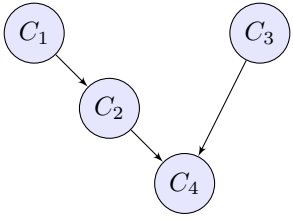
		
	Each poset can be drawn as a \textit{Hasse diagram}, which is a network whose nodes are the elements in the poset, and the edges obey the following rule: Let \(C_r \preceqdot C_s\) denote that \(C_s\) \textit{covers} \(C_r\), which means that \(C_r \preceq C_s\) such that \(r \neq s\) and there is no \(t\) (distinct from \(r\) and \(s\)) such that \(C_r \preceq C_t \preceq C_s\). In the poset's corresponding Hasse diagram, there is a directed edge from \(C_s\) to \(C_r\) if and only if \(C_r \preceqdot C_s\). Figure \ref{fig:RA-hasse} illustrates the Hasse diagram of the poset of strongly connected components obtained from the network in Figure \ref{fig:RA-smallnetwork}. In this Hasse diagram we have
	\begin{equation}
		\begin{aligned}
			C_4 &\preceqdot C_2, & C_2 &\preceqdot C_1, &C_4 &\preceqdot C_3,
		\end{aligned}
	\end{equation}
	but, for example, \(C_4 \not\preceqdot C_1\) since \(C_4 \preceq C_2 \preceq C_1\).
	
	Revisiting the RA model \eqref{eq:RA_update} on scalar form for nodes \(v_4\) and \(v_5\) in our example network, we have
	\begin{subequations}
		\begin{equation}
			x_{t+1,4} = (1-\alpha)x_{t,4} + \alpha w_{45}a_{t,5},
		\end{equation}
		\begin{equation}
			x_{t+1,5} = (1-\alpha)x_{t,5} + \alpha w_{54}a_{t,4},
		\end{equation}
	\end{subequations}
	from which we see that the update rule is independent from agents outside of the component, so clearly we can treat any question of convergence of the component \(C_3\) as if it was a disjoint component. Since this holds true for all maximal elements in a corresponding Hasse diagram for any network, we can apply Theorem \ref{thm:main} to each of these components independently.
	
	Let \(C_s\) be a maximal element, and let \(C_r \preceqdot C_s\). Then there is a least one agent \(k \in C_r\) and one agent \(l \in C_s\) such that \(w_{kl} \neq 0\); that is, there is a link from \(l\) to \(k\). By Lemma \ref{lem:union_corners}, all agents converge (collectively) towards either \(0\) or \(1\), and by the earlier discussion, the component \(C_s\) converges almost surely to the same opinion, say \(1\) (by which we mean that all agents in \(C_s\) converge almost surely to \(1\)).  By Lemma \ref{lem:twocorners} it follows that if agent \(k\) converges at all, it must also be to \(1\). This happens if \(C_s\) is the only maximal element. By  the irreducibility of \(C_r\), we can apply Lemma \ref{lem:twocorners} repeatedly to infer that all agents in \(C_r\) must converge to \(1\).
	
	If there is more than one maximal element, the network will converge to consensus in the sense of \eqref{eq:RA-convergence-definition} if and only if all maximal elements converge to the same value. To see this, consider the Hasse diagram in Figure \ref{fig:RA-hasse-general}, in which \(C_1\) and \(C_2\) are the only maximal elements. If \(C_1\) and \(C_2\) converge to different values, then, by the argument in the previous paragraph, \(C_4\) would have to converge towards both \(\matr{0}\) and \(\matr{1}\), but this is clearly impossible. On the other hand, if \(C_1\) and \(C_2\) converge to the same value, say \(\matr{1}\), then so will \(C_3, C_4\) and \(C_5\). Furthermore, since the minimal elements \(C_6\) and \(C_7\) are covered by \(C_3,C_4\) and \(C_5\), we can repeat the argument from the previous paragraph to show that also the minimal elements converge to \(\matr{1}\). This argument naturally generalizes to any poset \(\{C_1,\dots,C_M\}\). Hence, we have proved the following result.
	
	\begin{thm}\label{thm:RA-consensus-general-network}{\ \\}
	(i) Given a poset of irreducible components, each component corresponding to a maximal element converges almost surely to either \(\matr{0}\) or \(\matr{1}\) under the RA model. {\ \\} {\ \\}
	(ii) The RA model converges to consensus almost surely, i.e.,
		\begin{equation}\label{eq:RA-convergence-definition}
			\matr{x}_t \xrightarrow{a.s.} \matr{x}_\infty \in \{\matr{0},\matr{1}\},
		\end{equation}
		if and only if all maximal elements in the poset of irreducible components converge to the same value.
	\end{thm}
	\begin{figure}
		\centering
		\begin{tikzpicture}[auto,main_node/.style={circle,fill=blue!10,draw,minimum size=.4em,inner sep=1pt}, edge/.style={->,> = latex'}]
			
			\node[main_node] (1) at (-4.5,0)  {\(C_1\)};
			\node[main_node] (2) at (-2.5,0)  {\(C_2\)};
			\node[main_node] (3) at (-5.5,-1.5)  {\(C_3\)};
			\node[main_node] (4) at (-3.5,-1.5)  {\(C_4\)};
			\node[main_node] (5) at (-1.5,-1.5)  {\(C_5\)};
			\node[main_node] (6) at (-4.5,-3)  {\(C_6\)};
			\node[main_node] (7) at (-2.5,-3)  {\(C_7\)};
			
			\draw[edge] (1) to (3) {};
			\draw[edge] (1) to (4) {};
			\draw[edge] (2) to (4) {};
			\draw[edge] (2) to (5) {};
			\draw[edge] (3) to (6) {};
			\draw[edge] (4) to (6) {};
			\draw[edge] (4) to (7) {};
			\draw[edge] (5) to (7) {};
		\end{tikzpicture}%
		\caption[Short description]{Hasse diagram of the poset of strongly connected components of a (non-specified) network. Here, the maximal elements are \(C_1, C_2\), and the minimal elements are \(C_6, C_7\).}\label{fig:RA-hasse-general}
	\end{figure}
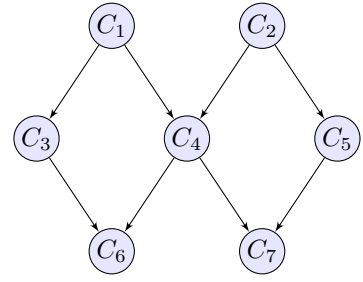
	
	As a special case, consider a network where there is exactly one maximal element in the poset of irreducible components. By Theorem \ref{thm:RA-consensus-general-network} for such a network the RA model converges to consensus almost surely. If the maximal element consists of a single agent, then this agent influences others but is not influenced by anyone else. Even if the initial belief of this agent is strictly between \(0\) and \(1\) it will converge almost surely to either \(0\) or \(1\) and the whole network will reach consensus to that value. In the remaining cases, when the initial belief of this agent (say agent \(k\)) is \(0\) or \(1\), the belief stays constant almost surely as \(x_{t+1,k} = (1-\alpha) x_{t,k} + \alpha a_{t,k}= x_{t,k}\), since  \(x_{t,k}=a_{t,k}\) with probability \(1\). Such an agent is known as a stubborn agent. We obtain the following by-product directly from Theorem \ref{thm:RA-consensus-general-network}.
	\begin{cor}\label{cor:RA-as-stubborn}
		The RA model converge to consensus in the sense of \eqref{eq:RA-convergence-definition} in the presence of a stubborn agent.
	\end{cor}
	While it may appear that this result would follow directly from available convergence results for related models with stubborn or extreme agents (notably \cite{yildiz13} and \cite{martins10}), it does not: Reference \cite{yildiz13} considered the voter model and showed only convergence in distribution, and \cite{martins10} dealt with the CODA model and contrarian agents which are not stubborn.

	Theorem \ref{thm:RA-consensus-general-network} together with a standard result from probability theory reveals that the RA model converges to consensus also in \(r\)th moment.
	\begin{cor}\label{cor:RA-moments} {\ \\}
		(i) The RA model converges to consensus in \(r\)th moment for every \(r > 0\), i.e.
		\begin{equation}\label{eq:ms-ordinary}
			\matr{x}_t \xrightarrow{L^r} \matr{x}_\infty \in \{\matr{0},\matr{1}\},
		\end{equation}
		if all maximal elements in the poset defined by the irreducible components in the network converge almost surely to the same value. {\ \\} {\ \\}
		(ii) The RA model converges to consensus in \(r\)th moment for every \(r > 0\) in the presence of a stubborn agent.
	\end{cor}
	\begin{proof}
		(i) By Theorem \ref{thm:RA-consensus-general-network}, \(\matr{x}_t \xrightarrow{a.s.} \matr{x}_\infty \in \{\matr{0},\matr{1}\}\). We also know, by definition of the RA model, that \(0 \leq x_{t,i} \leq 1\) for all \(i = 1,2,\dots,N\) and for all \(t \geq 1\), and therefore it holds that \(0 \leq x_{t,i}^r \leq 1\), for all \(r > 0\). By a standard result in probability theory (see, e.g., \cite[Theorem 2 (b)]{Ferguson}, whose proof essentially is an application of the Dominance Convergence Theorem) it follows that \(\matr{x}_t \xrightarrow{L^r} \matr{x}_\infty \in \{\matr{0},\matr{1}\}\).
		(ii) This is a special case of (i) in which the stubborn agent constitutes the only maximal element in the poset of strongly connected components.
	\end{proof}

	\subsection{Numerical examples}
		We illustrate Theorem 2 with two numerical examples, both featuring the reducible network in Figure \ref{fig:RA-smallnetwork}. For any given agent, the incoming edges have equal weights that sum to one. In Figure \ref{fig:conv-numerical}, the network converges to consensus under the RA model since all maximal elements in the corresponding Hasse diagram tend to the same value (namely 1). In Figure \ref{fig:non-conv-numerical}, however, the maximal elements tend to different values, so the dynamics model does not converge.
	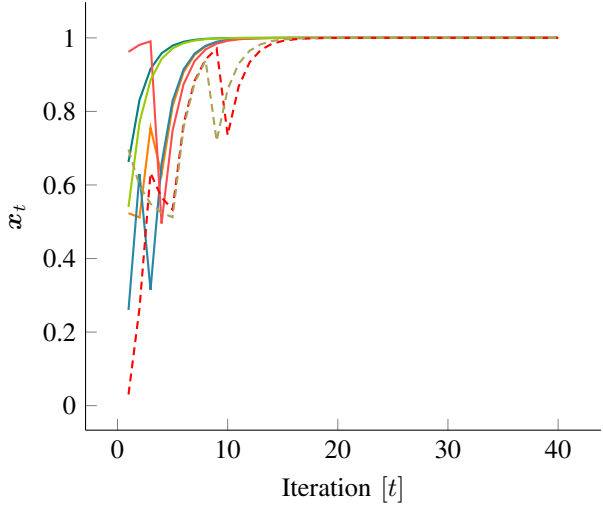
\begin{figure}
		\centering
		\begin{tikzpicture}[scale=1]
			\begin{axis}[
				axis lines*=left,
				cycle list name=exotic,
				xlabel={Iteration \([t]\)},
				ylabel={\(\matr{x}_t\)},
				ytick = {0,0.2,0.4,0.6,0.8,1},
				yticklabels = {0,0.2,0.4,0.6,0.8,1},
				xtick =  {0,10, ..., 40},
				xticklabels={0,10, ..., 40}
				]
				\pgfplotstableread[col sep=comma]{example_convergence.dat}\datatable
				\addplot +[thick,mark = none] table[x index=0, y index = 1] {\datatable};
				\addplot +[thick,mark = none] table[x index=0, y index=2] {\datatable};
				\addplot +[thick, mark = none] table[x index=0, y index=3] {\datatable};
				\addplot +[thick, mark = none] table[x index=0, y index=4] {\datatable};
				\addplot +[thick, mark = none] table[x index=0, y index=5] {\datatable};
				\addplot +[thick, mark = none] table[x index=0, y index=6] {\datatable};
				\addplot +[thick, mark = none] table[x index=0, y index=7] {\datatable};
			\end{axis}
		\end{tikzpicture}
		\caption{An example in which the RA model converges for a reducible network. Both of the maximal elements in the Hasse diagram (see Figure \ref{fig:RA-hasse}) tend to 1. The nodes \(v_6\) and \(v_7\) are marked with dashed lines.}
		\label{fig:conv-numerical}
	\end{figure}

	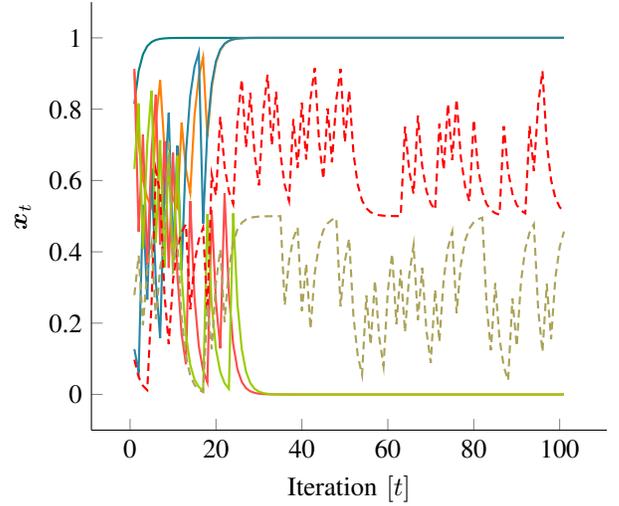
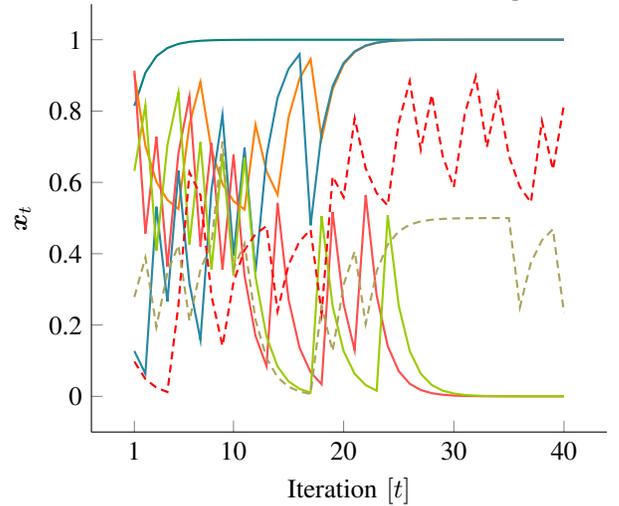
\begin{figure}[h!]
		\centering
		\begin{subfigure}{0.5\textwidth}
			\centering

		\begin{tikzpicture}[scale=1]
			\begin{axis}[
				axis lines*=left,
				cycle list name=exotic,
				xlabel={Iteration \([t]\)},
				ylabel={\(\matr{x}_t\)},
				ytick = {0,0.2,0.4,0.6,0.8,1},
				yticklabels = {0,0.2,0.4,0.6,0.8,1},
				xtick =  {0,20,...,100},
				xticklabels={0,20,...,100}
				]
				\pgfplotstableread[col sep=comma]{example_divergence.dat}\datatable
				\addplot +[thick,mark = none] table[x index=0, y index = 1] {\datatable};
				\addplot +[thick,mark = none] table[x index=0, y index=2] {\datatable};
				\addplot +[thick, mark = none] table[x index=0, y index=3] {\datatable};
				\addplot +[thick, mark = none] table[x index=0, y index=4] {\datatable};
				\addplot +[thick, mark = none] table[x index=0, y index=5] {\datatable};
				\addplot +[thick, mark = none] table[x index=0, y index=6] {\datatable};
				\addplot +[thick, mark = none] table[x index=0, y index=7] {\datatable};
			\end{axis}
		\end{tikzpicture}
		\caption{The  RA model does not  converge to consensus. The agents belonging to the strongly connected components \(C_1\) and \(C_2\) (see Figure \ref{fig:RA-hasse}) tend to 1, and the agents belonging to \(C_3\) tend to 0. Since \(C_1\) and \(C_3\) are maximal elements in the Hasse diagram, the nodes \(v_6\) and \(v_7\) (dashed lines), constituting the minimal element \(C_4\), are doomed to fluctuate between these two extreme points.}\label{fig:non-conv-example}
		\end{subfigure}
		\hfill
		\begin{subfigure}{0.5\textwidth}
		\centering
		\begin{tikzpicture}[scale=1]
		\begin{axis}[
			axis lines*=left,
			cycle list name=exotic,
			xlabel={Iteration \([t]\)},
			ylabel={\(\matr{x}_t\)},
			ytick = {0,0.2,0.4,0.6,0.8,1},
			yticklabels = {0,0.2,0.4,0.6,0.8,1},
			xtick =  {1,10,20,30,40},
			xticklabels={1,10,20,30,40}
			]
			\pgfplotstableread[col sep=comma]{example_divergence_zoom.dat}\datatable
			\addplot +[thick,mark = none] table[x index=0, y index = 1] {\datatable};
			\addplot +[thick,mark = none] table[x index=0, y index=2] {\datatable};
			\addplot +[thick, mark = none] table[x index=0, y index=3] {\datatable};
			\addplot +[thick, mark = none] table[x index=0, y index=4] {\datatable};
			\addplot +[thick, mark = none] table[x index=0, y index=5] {\datatable};
			\addplot +[thick, mark = none] table[x index=0, y index=6] {\datatable};
			\addplot +[thick, mark = none] table[x index=0, y index=7] {\datatable};
		\end{axis}
	\end{tikzpicture}
	\caption{A zoomed-in version of Figure \ref{fig:non-conv-example}.}\label{fig:non-conv-example-zoomed}
		\end{subfigure}
	\caption{An example in which the RA model does not converge.}\label{fig:non-conv-numerical}
\end{figure}

	\section{A critique of the proof of Theorem 1 in \cite{scaglione}}\label{sec:critique}
	In \cite[Theorem 1]{scaglione}, it is claimed that the RA model described by \eqref{eq:RA_update} leads to herding, in the sense that
	\begin{equation}\label{eq:scaglioneclaim}
	  \mathbb{P} \paren{\lim_{t \to \infty}x_{t,i} \in \{0,1\}} = 1, \ \text{for all } i=1,2,\dots N,
	\end{equation}
	and moreover that the limit is identical to all agents.\footnote{With our notation, which will be used throughout this section.} This is the same claim as our Theorem~\ref{thm:main}.	
	The main steps of the proof of \cite[Theorem 1]{scaglione} are: 
	\begin{enumerate}
	\item First, the random variable \(q_t =\matr{ \matr{\pi}}^T \matr{x}_t\) is defined, where \(\matr{\pi}\) is a left eigenvector of \(\matr{W}\) with eigenvalue \(1\), and it is shown that \(q_t\) is a martingale with respect to \(\matr{x}_t\), i.e.,
	\begin{equation}
	\E\{q_{t+1}\mid \matr{x}_t\} = q_t.
	\end{equation}
	\item Then the martingale difference sequence 
	\begin{equation}
		\mathrm{\Delta}q_t = q_t - q_{t-1}
	\end{equation}
	is shown to satisfy \(\mathrm{\Delta}q_t \xrightarrow{a.s} 0\), as \(t\to\infty\).
	\item The almost sure convergence in 2) is used to show that \(\mathrm{\Delta} q_t\) converges in the mean square sense, i.e.,  \(\E \{(\mathrm{\Delta}q_t)^2\}\to 0\), as \(t\to\infty\). Up to this point in the proof we are in agreement with all arguments.
	 \end{enumerate}
	 
	 Our main concern with the proof is the following.
	 The proof claims that 
	 \begin{quotation}
	 	``since for all \(i\), \(\pi_i > 0\), the MS convergence implies that
	 	\begin{equation}\label{eq:MS-leads-to-claim}
	 		\lim_{t\to\infty}x_{t,i}(1-x_{t,i}) = 0, \quad \forall i.\text{''}
	 	\end{equation}
	 \end{quotation}
	  It is not clear in what sense one should understand the convergence in \eqref{eq:MS-leads-to-claim}.
	  Clearly, convergence holds in the mean square sense: 
	Similar to our result in \cite[Equation (23)]{abrahamsson2019}, we can show that
		\begin{equation}\label{eq:MS-product}
			\lim_{t\to\infty}\E\{(x_{t,i}(1-x_{t,i}))^2\} = 0, \quad \forall n.
		\end{equation}
	However, for the convergence in \eqref{eq:MS-leads-to-claim} to be useful for subsequent arguments in the proof in \cite{scaglione},
	\eqref{eq:MS-leads-to-claim} must hold almost surely. More specifically,
	the convergence is later used (in \cite[Equation (18)]{scaglione}) to argue that
	\begin{equation}\label{eq:eitheror-claim}
	\lim\limits_{t\to\infty}x_{t,i} = 0 \text{ or } \lim\limits_{t\to\infty}x_{t,i} = 1.
	\end{equation}
	But \eqref{eq:MS-leads-to-claim}, interpreted in the sense of mean square convergence, does not  imply convergence
	in \eqref{eq:MS-leads-to-claim} almost surely, let alone does it imply \eqref{eq:eitheror-claim}. As a counterexample, consider a distribution which always results in the outcome
		\begin{equation}\label{eq:counterexample}
			x_{t,i} = \begin{cases}
				1, & t \text{ odd},\\
				0, & t \text{ even}.
			\end{cases}
		\end{equation}
	Then \(x_{t,i}(1-x_{t,i}) = 0\) for all \(t\), but \( \lim\limits_{t\to\infty}x_{t,i}\) does not exist, regardless of the mode of convergence. Note, however, that in the light of Theorem \ref{thm:main}, such a distribution is not compatible with the RA model: By the theorem there exists a time \(t_0\) such that, for all \(t \geq t_0\), the vector \(\matr{x}_t\) stays \(\delta\)-close to either of the vectors \(\matr{0}\) or \(\matr{1}\). Since this holds for any \(\delta\), the sequence in \eqref{eq:counterexample} does not lie in the sample space \(\Omega\). Hence, while \eqref{eq:counterexample} serves as a counterexample for the claim discussed in this section, it does not contradict the convergence result of the RA model.
	
	While the issue just explained constitutes our main point of criticism,
	we note in passing that \cite[Equation (16)]{scaglione} as written is inaccurate.
	That equation  states that
	\begin{equation}
	\begin{aligned}
	\Var (\mathrm{\Delta} q_t &\mid q_{t-1}) = \E\{(\mathrm{\Delta} q_t)^2 \mid q_{t-1}\}\\
	&= \alpha^2\sum_{i=1}^N \pi_i^2\Var(a_{t,i})\\
	&= \alpha^2 \sum_{i=1}^N\pi_i^2 x_{t,i}(1-x_{t,i}),
	\end{aligned}
	\end{equation}
	but should read
	\begin{equation}
	\begin{aligned}
	\Var(\mathrm{\Delta} q_t &\mid \matr{x}_{t-1}) = \E\{(\mathrm{\Delta} q_t)^2 \mid \matr{x}_{t-1}\}\\
	&= \alpha^2\sum_{i=1}^N \pi_i^2\Var(a_{t-1,i} \mid x_{t-1,i})\\
	&= \alpha^2 \sum_{i=1}^N\pi_i^2 x_{t-1,i}(1-x_{t-1,i}).
	\end{aligned}
	\end{equation}
	The corrections that should be applied to \cite[Equation (16)]{scaglione}  are the following:
	\begin{itemize}
		\item The conditional variance and conditional expectation on the first line should be with respect to \(\matr{x}_{t-1}\), not  \(q_{t-1}\). Otherwise one cannot make use of the definition \(q_t = \matr{\pi}^T\matr{x}_t\) to simplify the expression, and hence the subsequent equality would not hold.
		\item On the second line, the variance should be the conditional variance \(\Var(a_{t-1,i} \mid x_{t-1,i})\). Note also the time shift, which follows from the previous line.
		\item On the third line, the summand should be\\ \(\pi_i^2x_{t-1,i}(1-x_{t-1,i})\), i.e., once again the time variable should be shifted.
	\end{itemize}

	\section{Discussion}\label{sec:discussion}
	In this section we discuss the RA model in relation to two other questions. First, we consider what can be said if the update step sizes are time-variant. Second, we comment on the similarities and differences between the RA model and information propagation schemes in multi-agent optimization.
	\subsection{On the convergence of time-variant models}
	What happens if the step-size parameter \(\alpha\) is time-variant? Will the RA model still converge to consensus, and if so, under what conditions? While we do not have a firm answer to these questions, we can approach this discussion by studying two related examples; one deterministic model and one that involves randomness. Let us first consider a deterministic model. We will see that if the step-size decreases quickly enough, then the model might never reach a consensus.
	Let \begin{equation}\label{eq:simple_update_scheme}
		\matr{x}_{t+1} = \matr{W}_t \matr{x}_t,
	\end{equation}
	where, for some \(\beta_t \in (0,1/2]\),
	\begin{equation}
		\matr{W}_t = \begin{pmatrix}
			1 - \beta_t & \beta_t\\
			\beta_t & 1 - \beta_t
		\end{pmatrix}.
	\end{equation}
	Consider the eigenvalue decomposition of \(\matr{W}_t\),
	\begin{equation}
		\matr{W}_t = \matr{U}\matr{D}_t\matr{U}^T,
	\end{equation}
	where
	\begin{equation}
		\matr{U} = \frac{1}{\sqrt{2}}\begin{pmatrix}
			1 & -1\\
			1 & 1
		\end{pmatrix}, \qquad 
		\matr{D}_t =
		\begin{pmatrix}
			1 & 0\\
			0 & 1-2\beta_t
		\end{pmatrix}.
	\end{equation}
	Furthermore,
	\begin{equation}
		\matr{x}_T =\paren{ \prod_{t = 1}^{T} \matr{W}_t} \matr{x}_0 = \matr{U}\paren{\prod_{t = 1}^{T} \matr{D}_t} \matr{U}^T \matr{x}_0.
	\end{equation}
	In the time-invariant case we have 
	\begin{equation}
		\beta_t = \beta, \qquad \matr{D}_t = \matr{D} = 
		\begin{pmatrix}
			1 & 0\\
			0 & 1-2\beta
		\end{pmatrix},
	\end{equation}
	so \begin{equation}
			\lim_{T \to \infty}  \matr{U}\paren{\prod_{t = 1}^{T} \matr{D}_t} \matr{U}^T = \matr{U}
			\begin{pmatrix}
				1 & 0 \\
				0 & 0
			\end{pmatrix}
			\matr{U}^T = \frac{1}{2}\matr{1}\matr{1}^T, 
		\end{equation}
		which implies that
		\begin{equation}
			\lim_{T \to \infty} \matr{x}_T = \frac{1}{2}\matr{1}\matr{1}^T\matr{x}_0.
	\end{equation}
	
	However, suppose that \(\beta_t = \dfrac{1}{2^t}\beta\) (say). Then
	\begin{equation}
		\matr{D}_t = 
		\begin{pmatrix}
			1 & 0 \\
			0 & 1-\frac{\beta}{2^{t-1}}
		\end{pmatrix}.
	\end{equation}
	We now apply Lemma \ref{lem:prod-cont-decr}, with \(\alpha = \frac{1}{2}, N = 1\) and \(\gamma = 2\beta\), by which we obtain
	\begin{equation}
		  \prod_{t=0}^\infty \paren{1-\frac{\beta}{2^{t-1}}} = \prod_{t=0}^\infty \paren{1-\paren{\frac{1}{2}}^t 2\beta} = g_{\frac{1}{2},1}(2\beta) > 0.
	\end{equation}
	Therefore, in the limit as \(T \to \infty\), the product \(\prod_{t=1}^{T}\matr{D}_t\) converges to a diagonal matrix for which both the diagonal entries are positive. Hence the network will not reach consensus (except if \(\matr{x}_0 = \lambda \matr{1}\) for some constant \(\lambda \in \reals\)). 
	
	Now, let us turn our discussion to the stochastic model in \cite{tabhaz-salehi08}, in which the update schemes is the same as in \eqref{eq:simple_update_scheme}, but the weight matrices \(\{\matr{W}_t\}_{t=1}^\infty\) are i.i.d. stochastic matrices with positive diagonal elements. The authors show that the model reaches consensus (although not an average consensus since the limit is random) asymptotically under fairly mild assumptions. Specifically, consensus is reached if and only if the second largest eigenvalue (in magnitude) of \(\E\{\matr{W}_t\}\) is strictly less than \(1\).
	
	Hence, we have considered a deterministic model in which consensus was not reached, and a stochastic model in which it was.	With these two examples in mind, we do not believe that there is a direct time-variant variation of the RA model for which convergence to consensus is easily established.
	
	\subsection{The RA model in relation to multi-agent optimization}
	In distributed optimization, convergence to consensus is often necessary for the optimization scheme to work \cite{yang19,koppel15,terelius11,boyd11}. For example, this is true for consensus problems \cite{konstantinos12}, where a network of agents must achieve a consistent opinion by exchanging information locally; applications include vehicle formation and load balancing (see, e.g., \cite[Section IA]{dimakis10}). Hence it might seem natural that the RA model and other CODA-like models should be very closely related to multi-agent optimization algorithms, since they all are concerned with convergence to consensus. However, in opinion dynamics, reaching consensus is not necessarily desirable: In some contexts, such as political elections, consensus is unlikely to be observed empirically. In fact, as remarked in the introduction, the CODA-like models were designed to prevent the guarantee of reaching consensus, with the ambition of being more realistic. The RA model serves the purpose of showing that a CODA model with a simple probabilistic interaction fails to prevent convergence to consensus. Therefore, while many of the mathematical models and techniques used in opinion dynamics might be similar to those used in distributed optimization, the two research fields are often conceptually very different.
	
	\section{Conclusions}\label{sec:openproblems}
	In this paper we have shown that the Random Actions model converges to consensus almost surely and in \(r\)th moments, and that this holds also in the the presence of a stubborn agent. More generally, we showed that the assumption of strong connectedness (irreducibility) is sufficient but not necessary for the convergence result. While Theorem \ref{thm:main} is equivalent to the statement of \cite[Theorem 1]{scaglione}, our critique in Section \ref{sec:critique} casts doubts on the validity of the proof. From \eqref{eq:MS-product} it is clear that the product \(x_ {t,i}(1-x_{t,i})\) converges to zero in mean square, but by the counterexample \eqref{eq:counterexample} this does not necessarily mean that each factor converges to zero in mean square. We have shown that the counterexample is not compatible with the RA model, but this is not clear in \cite{scaglione}.
	
	We now comment on the convergence result itself. 
	Opinion dynamics models with a linear update scheme are often expected to reach consensus. This, however, is rarely observed empirically (cf. Abelson's ``community cleavage problem'', discussed in Section I), and therefore consensus-reaching models are not likely to realistically capture the exchange of opinions.  In order to rectify this issue, CODA-like models were introduced. For this reason it is interesting that the RA model, being a modified CODA model which incorporates a specific form of randomness, again leads to consensus. One way to avoid reaching consensus is to introduce a nonlinearity in the model. For example, in \cite[Section III-A]{scaglione} the RA model was modified to a Hegselmann-Krausse-like update scheme, so that an agent only updates his opinion if the observations of the neighbors' actions are sufficiently close to his current opinion. The authors showed that in such a model there is no guarantee for convergence; the opinions might fluctuate indefinitely.
	
	The main takeaway from our paper is the following: It is the randomness of the observed actions that drives the herding behavior, and the limiting vector of actions (\(\matr{0}\) or \(\matr{1}\)) is unpredictable even when the initial opinion distribution \(\matr{x}_1\), the parameter \(\alpha\) and the edge weights are known. That is a clear difference from other linear update schemes but which are deterministic, like the DeGroot model \cite{degroot}, where it is sufficient to know the initial values and edge weights to predict the limiting vector of opinions. As remarked in \cite{scaglione}, this suggests that the presence of random effects in opinion dynamics can constitute a significant hurdle in the analysis of long-term behavior in social networks.

	\appendix
	In order to prove Lemma \ref{lem:conv_bti}, we need a bit of preparation. First we construct a martingale \(q_t\) and prove some useful properties that it exhibits. The martingale construction is the same as in \cite{scaglione}, but in our proof it is subsequently used in a different manner. We then construct a new random variable \(\matr{y}_t\) based on the martingale difference \(q_{t+1} - q_t\), and show that it converges to \(0\) in the mean square sense, i.e., \(\matr{y}_t \xrightarrow{m.s.} 0\). Finally, this is used in the succeeding proofs. We begin by defining
	\begin{equation}
	q_t = \matr{\pi}^T\matr{x}_t = \sum_{i=1}^N \matr{\pi}_ix_{t,i},
	\end{equation}
	and show that \(q_t\) is a martingale w.r.t. \(\matr{x}_t\). We have
	\begin{equation}\label{eq:martingale}
	\begin{aligned}
	\E\{q_{t+1} \mid \matr{x}_t\} &= \E\{\matr{\pi}^T\matr{x}_{t+1} \mid \matr{x}_t\}\\
	&= \E\left\{\matr{\pi}^T\Big((1-\alpha)\matr{x}_t + \alpha \matr{W}\matr{a}_t\Big) \mid \matr{x}_t\right\}\\
	&= (1-\alpha)\matr{\pi}^T\matr{x}_t + \alpha\matr{\pi}^T\E\{\matr{a}_t \mid \matr{x}_t\}\\
	&= (1-\alpha)q_t + \alpha\matr{\pi}^T\matr{x}_t\\
	&= (1-\alpha)q_t + \alpha q_t = q_t,
	\end{aligned}
	\end{equation}
	and further, for all \(t \geq 1\),
	\begin{equation}\label{eq:martingale_bound}
		\E \{ q_t\} =  \E \{\matr{\pi}^T \matr{x}_t\} \leq 1,
	\end{equation}
	so \(q_t\) is indeed a martingale w.r.t. \(\matr{x}_t\). Since \(\E\{q_t\}\) is bounded by a constant for each \(t\), it follows from the martingale convergence theorem \cite[Section 12.3]{Grimmett} that
	\begin{equation}\label{eq:martingale-as-convergence}
		q_t \xrightarrow{\text{a.s.}}q_\infty, \quad t\to\infty,
	\end{equation}
	for some random variable \(q_\infty\). This implies that for the martingale difference sequence, \(\mathrm{\Delta} q_t = q_{t+1} - q_t\), we have
	\begin{equation}\label{eq:mds-as-conv}
		\mathrm{\Delta} q_t \xrightarrow{\text{a.s.}} 0, \quad t\to\infty.
	\end{equation}
	It is also useful to derive an alternative expression for the martingale difference sequence:
	\begin{equation}\label{eq:deltaq}
	\begin{aligned}
	\mathrm{\Delta} q_t &= q_{t+1} - q_t = \matr{\pi}^T(\matr{x}_{t+1} - \matr{x}_t)\\
	&= \matr{\pi}^T\paren{(1-\alpha)\matr{x}_t + \alpha W\matr{a}_t - \matr{x}_t}\\
	&= \alpha\matr{\pi}^T (W\matr{a}_t - \matr{x}_t) = \alpha\matr{\pi}^T (\matr{a}_t - \matr{x}_t).
	\end{aligned}
	\end{equation}
	
	Define \(\matr{y}_t = \matr{a}_t-\matr{x}_t\). Then we have \(\mathrm{\Delta} q_t = \alpha \matr{\pi}^T\matr{y}_t\), and since \(\alpha \neq 0\), it follows by \eqref{eq:mds-as-conv} that
	\begin{equation}
	 \matr{\pi}^T\matr{y}_t \xrightarrow{\text{a.s.}} 0, \quad t\to\infty.
	\end{equation}
	Let \(f_{x_{t,i}}(x)\) and \(f_{y_{t,i}}(y)\)  denote the probability density functions of  \(x_{t,i}\) and \(y_{t,i}\), respectively. Then
	\begin{equation}
	\begin{aligned}
	f_{y_{t,i}}(y) &= (1+y)f_{x_{t,i}}(-y) + (1-y)f_{x_{t,i}}(1-y)\\
	&=
	\begin{cases}
	(1-y)f_{x_{t,i}}(1-y), & 0 < y \leq 1 \\
	(1+y)f_{x_{t,i}}(-y), & -1 \leq y < 0\\
	f_{x_{t,i}}(0) + f_{x_{t,i}}(1), & y = 0.
	\end{cases}
	\end{aligned}
	\end{equation}
	Observe that since \(0 \leq x_{t,i} \leq 1\), we have \(|y_{t,i}| \leq 1\). Furthermore, if \(\delta \leq x_{t,i} \leq 1-\delta\) for some \(\delta \in (0,1/2)\), then \(|y_{t,i}| \geq \delta\). Similarly,
	\begin{equation}
	\begin{aligned}
	&|y_{t,i}| = |a_{t,i}-x_{t,i}| \geq \delta\\
	 \implies 
	&\begin{cases}
	a_{t,i}-x_{t,i} \geq \delta, \quad \text{or}\\
	a_{t,i}-x_{t,i} \leq -\delta
	\end{cases}\\
	 \implies
	&\begin{cases}
	x_{t,i} \leq a_{t,i} - \delta \leq 1-\delta, \quad \text{or}\\
	x_{t,i} \geq a_{t,i} + \delta \geq \delta,
	\end{cases}
	\end{aligned}
	\end{equation}
	so in fact, for all \(\delta \in (0,1/2)\), we have
	\begin{equation}\label{eq:bound_y_delta}
	\delta \leq x_{t,i} \leq 1-\delta \text{ if and only if } |y_{t,i}| \geq \delta.
	\end{equation}
	The following two lemmas establish that for each time step \(t\), the random variables \(y_{t,1},\dots,y_{t,N}\) are uncorrelated, and furthermore that \(\matr{y}_t \xrightarrow{m.s.} \matr{0}\).
	\begin{lem}\label{lem:uncorrelated}
		\(y_{t,i}\) and \(y_{t,j}\), \(1\leq i < j \leq N\), \(t = 1,2,\dots\), are uncorrelated.
	\end{lem}
			\begin{proof}
		First, we calculate \(\E\{y_{t,i}\}\):
		\begin{equation}
		\begin{aligned}
		\E\{y_{t,i}\} &= \E\{a_{t,i} - x_{t,i}\}\\
		&= \E\{\E\{a_{t,i} - x_{t,i}\mid x_{t,i}\}\}\\
		&= \E\{\E\{a_{t,i}\mid x_{t,i}\} - x_{t,i}\}\\
		&= \E\{x_{t,i} - x_{t,i}\} = \E\{0\} = 0.
		\end{aligned}
		\end{equation}
		Then we calculate the correlation \(\E\{y_{t,i}y_{t,j}\}\):
		\begin{equation}
		\begin{aligned}
		&\E\{y_{t,i}y_{t,j}\} = \E\{(a_{t,i} - x_{t,i})(a_{t,j} - x_{t,j})\}\\
		= {} & \E\{\E\{(a_{t,i} - x_{t,i})(a_{t,j} - x_{t,j})\mid x_{t,i},x_{t,j}\}\}\\
		= {} & \E\{\E\{a_{t,i}-x_{t,i}\mid x_{t,i}\}
		\E\{a_{t,j} - x_{t,j}\mid x_{t,j}\}\}\\
		= {} & \E\{0\} = 0 =  \E\{y_{t,i}\} \cdot \E\{y_{t,j}\}.
		\end{aligned}
		\end{equation}
	\end{proof}
	\begin{lem}\label{lem:msconv}
		\(\E\{y_{t,i}^2\} \to 0,\) as \(t\to\infty\), for \(i = 1,\dots,N\).
	\end{lem}
	\begin{proof}
		We have
		\begin{equation}
		\matr{\pi}^T \matr{y}_t \xrightarrow{\text{a.s.}} 0, \quad t\to\infty.
		\end{equation}
		Note that, under the additional observation that \(|\matr{\pi}^T \matr{y}_t| \leq 1\),  convergence almost surely implies convergence in the mean square sense (see, e.g., \cite[Theorem 2 (b)]{Ferguson}). Hence, it follows that
		\begin{equation}
		\E\{(\matr{\pi}^T\matr{y}_t)^2\} \to 0, \quad t\to\infty.
		\end{equation}
		Now we calculate
		\begin{equation}
		\begin{aligned}
		\E\{(\matr{\pi}^T\matr{y}_t)^2\} &= \E\left\{\paren{\sum_{i=1}^N \matr{\pi}_i y_{t,i}}^2 \right\}\\
		&= \E\left\{\sum_{i=1}^N \sum_{j=1}^N \matr{\pi}_i \matr{\pi}_j y_{t,i} y_{t,j} \right\}\\
		&= \sum_{i=1}^N \matr{\pi}_i^2 \E\{y_{t,i}\} + \sum_{i\neq j} \matr{\pi}_i \matr{\pi}_j \E\{y_{t,i}y_{t,j}\}\\
		&= \sum_{i=1}^N \matr{\pi}_i^2 \E\{y_{t,i}^2\}.
		\end{aligned}
		\end{equation}
		Since \(\matr{\pi} \succ \matr{0}\) and \(\E\{y_{t,i}^2\} \geq 0\) for \(i = 1,\dots,N\), it follows that \(\E\{y_{t,i}^2\} \to 0,\) as \(t\to\infty\), for \(i = 1,\dots,N\).
	\end{proof}
	We are now ready to prove the lemmas from Section \ref{sec:results}.
	\begin{proof}[Proof of Lemma \ref{lem:conv_bti}]
		Suppose that the opposite is true, and consider the complement, which we denote by \(\paren{\cdot}^\complement\). Then, for some \(i\), there exists an \(\epsilon > 0\) such that for any \(t\in\naturals\), there exists \(t_0 > t\) such that
		\begin{equation}
			\P\paren{\paren{C_{t_0,i}(\delta)}^\complement} = \P\paren{\delta \leq x_{t_0,i} \leq 1-\delta} > \epsilon,
		\end{equation}
		which in terms of \eqref{eq:bound_y_delta} can be expressed as
		\begin{equation}
			\P\paren{|y_{t_0,i}| \geq \delta} > \epsilon.
		\end{equation}
		We also have
		\begin{equation}
			\begin{aligned}
			 \E\{y_{t_0,i}^2\} &= \int_{-1}^1 y^2 f_{y_{t_0,i}}(y) dy\\
			 &\geq \int_{-1}^{-\delta} y^2 f_{y_{t_0,i}}(y) + \int_\delta^1 y^2 f_{y_{t_0,i}}(y) dy\\
			 &\geq \delta^2 \P\paren{|y_{t_0}| \geq \delta} \geq \epsilon \delta^2,
			\end{aligned}
		\end{equation}
		but this contradicts the result of Lemma \ref{lem:msconv}, which states that \(\E\{y_{t,i}^2\} \to 0\) as \(t~\to~\infty\).
	\end{proof}
	\hrule
	\vs
	\begin{proof}[Proof of Lemma \ref{lem:final}]
		For any \(\delta \in (0,1/2)\) and any \(t \in \naturals\), we have, by using the subadditive property of probability measures together with De Morgan's law,
		\begin{equation}
			\P\paren{C_t(\delta)} \geq \sum_{i=1}^N \P\paren{C_{t,i}(\delta)} - (N-1).
		\end{equation}
		By Lemma \ref{lem:conv_bti}, we obtain
		\begin{equation}
			\lim_{t \to \infty} \sum_{i=1}^N \P\paren{C_{t,i}(\delta)} = \sum_{i=1}^N  \lim_{t \to \infty}\P\paren{C_{t,i}(\delta)} = N,
		\end{equation}
		and therefore
		\begin{equation}
			\lim_{t \to \infty} \P\paren{C_t(\delta)} \geq \lim_{t \to \infty} \sum_{i=1}^N \P\paren{C_{t,i}(\delta)} - (N-1) = 1.
		\end{equation}
	\end{proof}
	\hrule
	\vs
	\begin{proof}[Proof of Lemma \ref{lem:twocorners}]
		We prove the first case, i.e., with \(C_{t,k}^{(0)}(\delta)\) and \(C_{t,l}^{(1)}(\delta)\), since the other case follows easily with symmetrical arguments. We want to show that
		\begin{equation}\label{eq:lemma5-bound}
			\P'\paren{h^{-1}\paren{C_{t,l}^{(1)}(\delta)} \cap \{\omega' \in \Omega': a_l(t) = 0\}} \leq \delta,
		\end{equation}
		and, for sufficiently large \(t\),
		\begin{equation}\label{eq:lemma5-limit}
			\P'\paren{h^{-1}\paren{C_{t,k}^{(0)}(\delta)} \cap \{\omega' \in \Omega': a_l(t) = 1\}} \leq \epsilon,
		\end{equation}
		so that
		\begin{equation}
		\begin{aligned}
			&\P \paren{C_{t,k}^{(0)}(\delta) \cap C_{t,l}^{(1)}(\delta)}\\
			= {} & \P' 	\paren{h^{-1}\paren{C_{t,k}^{(0)}(\delta) \cap C_{t,l}^{(1)}(\delta)}} \\
			= {} & \P' 	\paren{h^{-1}\paren{C_{t,k}^{(0)}(\delta)} \cap h^{-1}\paren{C_{t,l}^{(1)}(\delta)}}\\
			= {} & \P' 	\paren{h^{-1}\paren{C_{t,k}^{(0)}(\delta)} \cap h^{-1}\paren{C_{t,l}^{(1)}(\delta)} \cap \Omega'}\\
			= {} & \P'\left(h^{-1}\paren{C_{t,k}^{(0)}(\delta)} \cap h^{-1}\paren{C_{t,l}^{(1)}(\delta)}\right.\\
			&\quad\quad \left. \phantom{C_{t,k}^{(0)}(\delta)} \cap \{\omega' \in \Omega': a_l(t) = 0\}\right) \\
			& \quad + \P'\left(h^{-1}\paren{C_{t,k}^{(0)}(\delta)} \cap h^{-1}\paren{C_{t,l}^{(1)}(\delta)}\right. \\
				& \quad\quad \left. \phantom{C_{t,k}^{(0)}(\delta)} \cap \{\omega' \in \Omega': a_l(t) = 1\}\right)\\
			\leq {} & \P'\paren{h^{-1}\paren{C_{t,l}^{(1)}(\delta)} \cap \{\omega' \in \Omega': a_l(t) = 0\}} \\
			& \quad + \P'\paren{h^{-1}\paren{C_{t,k}^{(0)}(\delta)} \cap \{\omega' \in \Omega': a_l(t) = 1\}}\\
			\leq {} & \delta + \epsilon,
		\end{aligned}
		\end{equation}
		where we have used the fact that \(\{\omega' \in \Omega': a_l(t) = 0\}\) and \(\{\omega' \in \Omega': a_l(t) = 1\}\) constitute a partition of \(\Omega'\).
		
		The bound in \eqref{eq:lemma5-bound} is easy to see, since
		\begin{equation}
		\begin{aligned}
			&\P'\paren{h^{-1}\paren{C_{t,l}^{(1)}(\delta)} \cap \{\omega' \in \Omega' : a_l(t) = 0\}}\\
			&= \int_{1-\delta}^1 f_{x_l(t),a_l(t)}(x,0) dx\\
			&= \int_{1-\delta}^1 f_{a_l(t) \mid x_l(t)}(0\mid x)f_{x_l(t)}(x) dx\\
			&= \int_{1-\delta}^1 (1-x) f_{x_l(t)}(x) dx\\
			&\leq \delta \int_{1-\delta}^1 f_{x_l(t)}(x) dx \leq \delta.
		\end{aligned}
		\end{equation}
		For the bound in \eqref{eq:lemma5-limit}, note that if the event
		\begin{equation*}
			h^{-1}\paren{C_{t,k}^{(0)}(\delta)} \cap \{\omega' \in \Omega' : a_l(t) = 1\}
		\end{equation*}
		occurs, then \(0~\leq~x_k(t) < \delta\) and \(a_l(t) = 1\).
		Then we have, on the one hand,
		\begin{equation}
			x_k(t+1) = (1-\alpha)x_k(t) + \alpha\sum_{i=1}^N w_{ki}a_{i,t} \geq \alpha w_{kl} \geq \delta,
		\end{equation}
		and on the other hand,
		\begin{equation}
			x_k(t+1) = (1-\alpha)x_k(t) + \alpha\sum_{i=1}^N w_{ki}a_{i,t} \leq (1-\alpha)\delta + \alpha \leq 1-\delta,
		\end{equation}
		which results in \(x_k(t+1) \in [\delta,1-\delta]\). This bound, in turn, implies that
		\begin{equation}\label{eq:lem-twocorners-set-inclusion}
			\begin{aligned}
			&h^{-1}\paren{C_{t,k}^{(0)}(\delta)} \cap \{\omega' \in \Omega' : a_l(t) = 1\}\\
			\subseteq {} & h^{-1}\paren{\paren{C_{t+1,k}(\delta)}^\complement} \subseteq h^{-1}\paren{\paren{C_{t+1}(\delta)}^\complement}.
			\end{aligned}
		\end{equation}
		We observe that \(\delta \leq \frac{1-\alpha}{2-\alpha} < \frac{1}{2}\), so we can apply Lemma \ref{lem:final}, by which we have
		\begin{equation}
			\lim_{t\to\infty} \P(C_{t+1,k}(\delta))  = 1,
		\end{equation}
		or, equivalently, by taking the complement,
		\begin{equation}\label{eq:B_ti_bound}
			\lim_{t\to\infty} \P\paren{\paren{C_{t+1,k}(\delta)}^\complement} = 0.
		\end{equation}
		By combining \eqref{eq:B_ti_bound} with \eqref{eq:lem-twocorners-set-inclusion} we conclude that
		\begin{equation}
		\begin{aligned}
			0 \leq &\lim_{t\to\infty} \P'\paren{h^{-1}\paren{C_{t,k}^{(0)}(\delta)} \cap \{\omega' \in \Omega' : a_l(t) = 1\}}\\
			\leq &\lim_{t\to\infty} \P' \paren{h^{-1}\paren{\paren{C_{t+1}(\delta)}^\complement}} = 0.
		\end{aligned}
		\end{equation}
		This is equivalent to saying that for all \(\epsilon > 0\), there exists a time \(t_\epsilon\)  such that for all \(t \geq t_\epsilon\),
		\begin{equation}
			\P'\paren{h^{-1}\paren{C_{t,k}^{(0)}(\delta)} \cap \{\omega' \in \Omega' : a_l(t) = 1\}} \leq \epsilon.
		\end{equation}
	\end{proof}
	\hrule
	\vs
	\begin{proof}[Proof of Lemma \ref{lem:union_corners}]
		Ultimately, we want to make statements about under what conditions \(\matr{x}_t\) is \(\delta\)-close to the corners \(\mathbf{0}\) and \(\mathbf{1}\). With this in mind, we will consider all other corners first. That is, we will assume \(\matr{m} \not\in \{\mathbf{0},\mathbf{1}\}\). We partition the set of agents into two disjunct subsets:
		\begin{equation}
			\{1,2,\dots,N\} = I_0(\matr{m}) \cup I_1(\matr{m}),
		\end{equation}
		where \(I_a(\matr{m})\) is defined by
		\begin{equation}
			I_a(\matr{m}) = \{i \in \{1,\dots,N\} : m_i = a\}, \quad a\in\{0,1\},
		\end{equation}
		for each \(\matr{m} \in \{0,1\}^N \setminus \{\mathbf{0},\mathbf{1}\}\).\footnote{For example, if \(N = 3\) and \(\matr{m} = (1,0,1)\), then \(I_0(\matr{m}) = \{2\}\) and \(I_1(\matr{m}) = \{1,3\}\).} Note that
		\begin{equation}
			I_0(\matr{m}) \neq \emptyset, \qquad I_1(\matr{m}) \neq \emptyset,
		\end{equation}
		and define
		\begin{equation}
			w(\matr{m}) = \max_{\substack{k,l \\ k \in I_0(\matr{m}) \\ l \in I_1(\matr{m})}} w_{kl}.
		\end{equation}
		Since \(\matr{W}\) is irreducible, it follows that \(w(\matr{m}) > 0\).
		\vs
		
		Now, note that if we prove the lemma for \(\delta_0 \in (0,\delta)\), then it also holds for \(\delta\). Therefore, w.l.o.g., let 
		\begin{equation}
			\delta_0 = \min \left\{\min_{\matr{m}}\alpha w(\matr{m}),\frac{1-\alpha}{2-\alpha}, \frac{\epsilon}{2(2^N-2)}, \delta\right\},
		\end{equation}
		and let
		\begin{equation}
			\epsilon_0 = \frac{\epsilon}{2(2^N-1)}.
		\end{equation}
		For fixed \(\matr{m}\), choose \(k,l\) such that \(w_{kl} = w(\matr{m})\). Note that
		\begin{equation}
			C_t^{\matr{m}}(\delta_0) = \bigcap_{i=1}^N C_{t,i}^{(m_i)}(\delta_0) \subseteq C_{t,k}^{(0)}(\delta_0) \cap C_{t,l}^{(1)}(\delta_0),
		\end{equation}
		which, together with Lemma \ref{lem:twocorners}, means that there exists a time \(t_{\epsilon_0}\) such that
		\begin{equation}\label{eq:lemma6-bound}
			\P\paren{C_t^{\matr{m}}(\delta_0)} \leq \P\paren{C_{t,k}^{(0)}(\delta_0) \cap C_{t,l}^{(1)}(\delta_0)} \leq \delta_0 + \epsilon_0
		\end{equation}
			for all \(t \geq t_{\epsilon_0}\) and \(\matr{m} \in \{0,1\}^N\). Consequently,
		\begin{equation}\label{eq:lemma6-finalform}
		\begin{aligned}
			&\P\paren{C_t(\delta_0)} = \P\paren{\bigcup_{\matr{m}\in \{0,1\}^N} C_t^{\matr{m}}(\delta_0)}\\
			= {} & \sum_{\matr{m} \in \{0,1\}^N} \P\paren{C_t^{\matr{m}}(\delta_0)}\\
			= {} & \P\paren{C_t^{\matr{0}}(\delta_0)} + \P\paren{C_t^{\matr{1}}(\delta_0)}
			+ \sum_{\substack{\matr{m} \in \{0,1\}^N \\ \matr{m} \neq \{\mathbf{0},\mathbf{1}\}}} \P\paren{C_t^{\matr{m}}(\delta_0)}\\
			\leq {} & \P\paren{C_t^{\matr{0}}(\delta_0)} + \P\paren{C_t^{\matr{1}}(\delta_0)} + (2^N-2)(\delta_0 + \epsilon_0),
		\end{aligned}
		\end{equation}
		where the second equality follows from \(\delta_0 \leq \frac{1-\alpha}{2-\alpha} < \frac{1}{2}\), so that the union is over disjoint sets. We apply Lemma \ref{lem:final}, by which there exists a time \(t_{\delta_0}\) such that
		\begin{equation}\label{eq:lemma6-bound2}
			\P\paren{C_t(\delta_0)} > 1 - \epsilon_0
		\end{equation}
		whenever \(t \geq t_{\delta_0}\). Let \(t_{\epsilon,\delta} = \max\{t_{\epsilon_0},t_{\delta_0}\}\) so that \eqref{eq:lemma6-bound} and \eqref{eq:lemma6-bound2} holds simultaneously for all \(t \geq t_{\epsilon,\delta}\). In view of \eqref{eq:lemma6-finalform}, we obtain, for all such \(t\),
		\begin{equation}
		\begin{aligned}
			&\P\paren{C_t^{\matr{0}}(\delta_0)} + \P\paren{C_t^{\matr{1}}(\delta_0)}\\
			\geq {} & \P\paren{C_t(\delta_0)} - (2^N-2)(\delta_0 + \epsilon_0)\\
			> {} & 1 - \epsilon_0 - (2^N-2)(\delta_0 + \epsilon_0)\\
			= {} &  1 - (2^N-2)\delta_0 - (2^N-1)\epsilon_0\\
			\geq {} & 1 - \frac{\epsilon}{2} - \frac{\epsilon}{2} = 1 - \epsilon.
		\end{aligned}
		\end{equation}
	\end{proof}
	\hrule
	\vs
	
	\begin{proof}[Proof of Lemma \ref{lem:prod-cont-decr}]
		For all \(\alpha \in (0,1)\), \(\gamma \in [0,1]\) and \(s \in \naturals\) we have 
		\begin{equation}\label{eq:fs-arg-bound}
			0<1-(1-\alpha)^s \gamma\leq 1.
		\end{equation}
		Hence, for all \(s \in \naturals\), each of the functions
		\begin{equation}
			f_s(\gamma)=\ln(1-(1-\alpha)^s \gamma)
		\end{equation}
		is defined, continuous and decreasing on \([0,1]\), and the infinite product can be written as 
		\begin{equation}\label{eq:log-def}
			\prod_{s=0}^{\infty} (1-(1-\alpha)^s \gamma)^N = N(1-\gamma )^N \exp\left( \sum_{s=1}^{\infty}f_s(\gamma)\right).
		\end{equation}
		For all \(\alpha \in (0,1), \gamma \in [0,1]\) and \(s \in \naturals\), it is easy to show that  \eqref{eq:fs-arg-bound} implies
		\begin{equation}\label{eq:help-lowerbound}
			\frac{(1-\alpha)^s\gamma }{ (1-\alpha)^s\gamma-1}\geq 
			-\frac{\gamma}{\alpha}(1-\alpha)^s. 
		\end{equation}
		Using \eqref{eq:help-lowerbound} and the inequality \(\ln(x) \geq 1 - 1/x\), which is valid for all positive \(x\), we can bound the functions \(f_s(\gamma)\) on the interval \([0,1]\) as follows:
		\begin{equation}
		\begin{aligned}
			0 \geq {} & f_s(\gamma) = \ln(1-(1-\alpha)^s \gamma)\\
			\geq {} & 1 - \frac{1}{1-(1-\alpha)^s\gamma}
			= \frac{(1-\alpha)^s\gamma}{(1-\alpha)^s\gamma -1}\\
			\geq {} & -\frac{\gamma}{\alpha}(1-\alpha)^s. 
		\end{aligned}
		\end{equation}
		Note that we have
		\begin{equation*}
			- \frac{\gamma}{\alpha}\sum_{s=1}^\infty (1-\alpha)^s = -\frac{\gamma (1-\alpha)}{\alpha^2},
		\end{equation*}
		and thus, by Weierstrass' M-test \cite[Theorem 7.10]{Rudin}, together with the uniform limit theorem \cite[Theorem 7.12]{Rudin} and the fact that each of the functions \(f_s(\gamma)\) is continuous and decreasing on the interval \([0,1]\), it follows that 
		\begin{equation}
			\sum_{s=1}^\infty f_s(\gamma) = f(\gamma),
		\end{equation}
		where \(f(\gamma)\) is a continuous and decreasing function on \([0,1]\).
		As the exponential function is continuous and strictly increasing, from \eqref{eq:log-def} we obtain the identity
		\begin{equation}
			\prod_{s=0}^{\infty} (1-(1-\alpha)^s \gamma)^N = (1-\gamma)^N \exp\left( f(\gamma)\right) = g_{\alpha,N}(\gamma),
		\end{equation}
		and since the nonnegative function \((1-\gamma)^N\) is continuous and decreasing on the interval \([0,1]\), so, too, is \(g_{\alpha,N}(\gamma)\). Finally, we observe that since \( f(1) \leq 0\) and \(f_s(0)=0\) for every \(s\in\naturals\), we have 
		\begin{equation}
			g_{\alpha,N}(0) = \exp\left(f(0)\right) = 1\mbox{ and } g_{\alpha,N}(1) = 0.
		\end{equation}
	\end{proof}
	\section*{Authors}
\begin{wrapfigure}{l}{25mm}
	\includegraphics[width=1in,height=1.25in,clip,keepaspectratio]{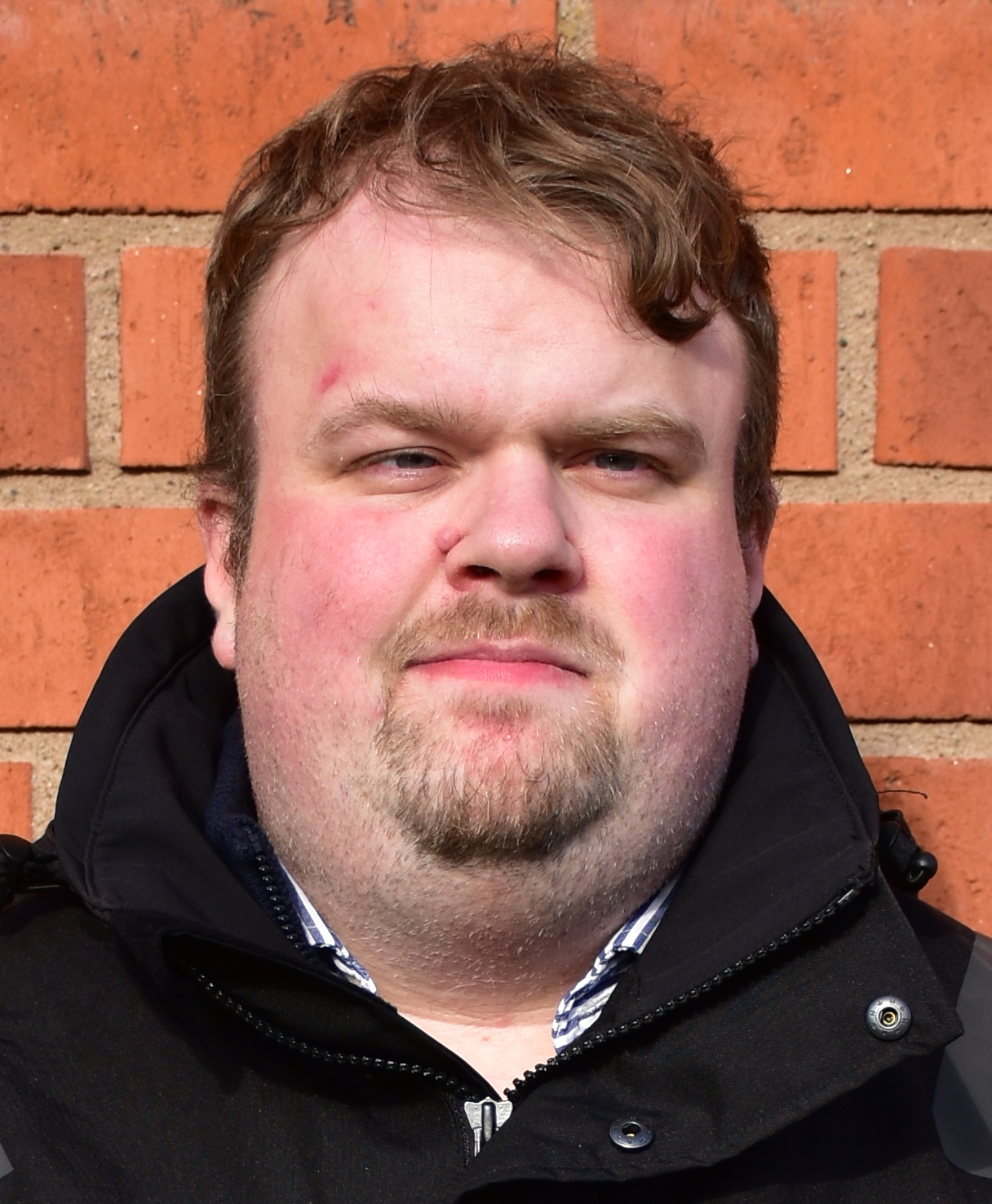}
\end{wrapfigure}\par
\textbf{Olle Abrahamsson} received his M.Sc. in Mathematics from Link\"oping University, Sweden, in 2017 and is currently pursuing his Licentiate degree in Electrical Engineering at Link\"oping University. His research interests are within the fields of network science and radar electronic warfare systems. He won the award for best poster presentation at the 2017 TAMSEC conference. Since 2023 he is a research scientist at the Swedish Defence Research Agency (FOI).

\begin{wrapfigure}{l}{25mm}
	\includegraphics[width=1in,height=1.25in,clip,keepaspectratio]{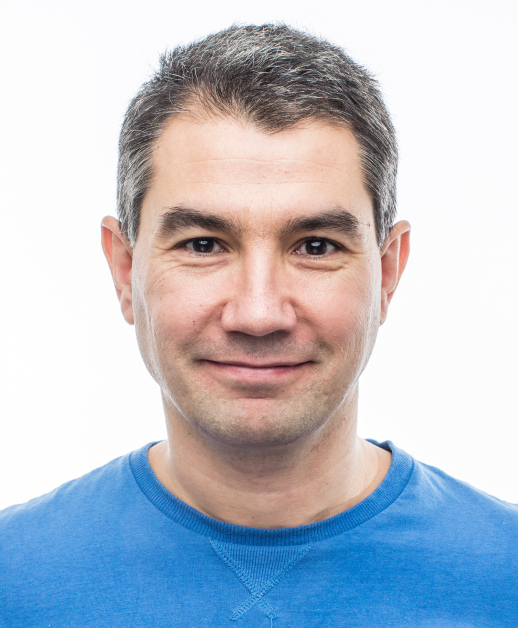}
\end{wrapfigure}\par
\textbf{Danyo Danev} received his M.Sc. in Mathematics from Sofia University, Bulgaria, in
1996 and his Ph.D. in Electrical Engineering from Link\"oping University, Sweden,
in 2001. In 2005 he obtained Docent title in Data Transmission. He is
currently Associate Professor at Linköping University.  His
research interests are within the fields of coding, information and
communication theory. He has authored or co-authored 2 book chapters,
16 journal papers and more than 30 conference papers on these
topics. He is currently teaching a number of communication
engineering and mathematics courses. Since 2012 he is board member
of the IEEE Sweden VT/COM/IT Chapter where from 2017 he is serving as treasurer.

\begin{wrapfigure}{l}{25mm}
	\includegraphics[width=1in,height=1.25in,clip,keepaspectratio]{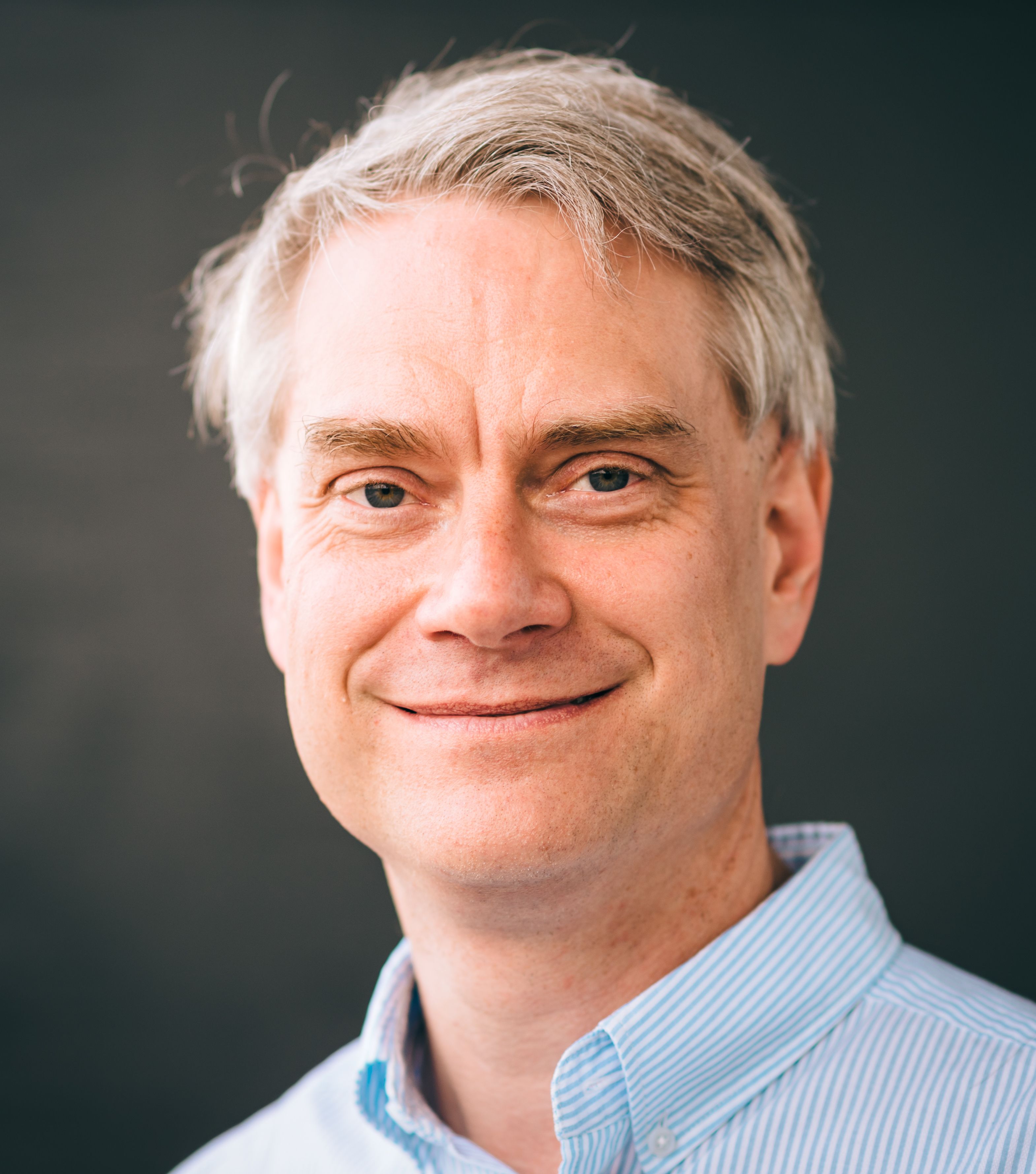}
\end{wrapfigure}\par
\textbf{Erik G. Larsson} (Fellow, IEEE) received the Ph.D. degree from Uppsala University,
Uppsala, Sweden, in 2002.  He is currently Professor of Communication
Systems at Link\"oping University (LiU) in Link\"oping, Sweden. He was
with the KTH Royal Institute of Technology in Stockholm, Sweden, the
George Washington University, USA, the University of Florida, USA, and
Ericsson Research, Sweden.  His main professional interests are within
the areas of wireless communications and signal processing. He 
co-authored \emph{Space-Time Block Coding for  Wireless Communications} (Cambridge University Press, 2003) 
and \emph{Fundamentals of Massive MIMO} (Cambridge University Press, 2016). 

He served as  chair  of the IEEE Signal Processing Society SPCOM technical committee (2015--2016), 
chair of  the \emph{IEEE Wireless  Communications Letters} steering committee (2014--2015), 
member of the  \emph{IEEE Transactions on Wireless Communications}    steering committee (2019-2022),
General and Technical Chair of the Asilomar SSC conference (2015, 2012), 
technical co-chair of the IEEE Communication Theory Workshop (2019), 
and   member of the  IEEE Signal Processing Society Awards Board (2017--2019).
He was Associate Editor for, among others, the \emph{IEEE Transactions on Communications} (2010-2014), 
the \emph{IEEE Transactions on Signal Processing} (2006-2010),
and  the \emph{IEEE Signal  Processing Magazine} (2018-2022).

He received the IEEE Signal Processing Magazine Best Column Award
twice, in 2012 and 2014, the IEEE ComSoc Stephen O. Rice Prize in
Communications Theory in 2015, the IEEE ComSoc Leonard G. Abraham
Prize in 2017, the IEEE ComSoc Best Tutorial Paper Award in 2018, 
the IEEE ComSoc Fred W. Ellersick Prize in 2019, and the
IEEE SPS Donald G. Fink Overview Paper Award in 2023.

He is a member of the Swedish Royal Academy of Sciences (KVA), and Highly Cited according to ISI Web of Science.
		\bibliographystyle{IEEEtran}
		\bibliography{IEEEabrv,RA-extended-paper} 
	\end{document}